\documentclass[a4paper,UKenglish]{lipics}
\usepackage{microtype}%if unwanted, comment out or use option "draft"

\def\calP{\mathcal{P}}

\def\calU{\mathcal{U}}
\def\calM{\mathcal{M}}

\def\calB{\mathcal{B}}

\def\st{$s$-$t$}
\def\c{core}

\newcommand{\TD}{\ensuremath{\mathcal{D}}}
\newcommand{\TDM}{\ensuremath{\mathcal{D}_{\!\calM}}}
\newcommand{\TDF}{\ensuremath{\mathcal{D}_{\!f}}}
\newcommand{\spm}{\ensuremath{S\!P\!M}}
\newcommand{\dmax}{\ensuremath{R}}
\newcommand{\funnel}{\ensuremath{F}}
\newcommand{\Tri}{\ensuremath{\mathcal{T}}}

\newcommand{\Real}{\ensuremath{\mathbb{R}}}

\newcommand{\bd}{\ensuremath{\partial}}

\newcommand{\seg}{\overline}

\newtheorem{fact}{Fact}

\title{Computing the $L_1$ Geodesic Diameter and Center of a Polygonal Domain\footnote{A preliminary version of this paper appeared in the Proceedings of the 33rd International Symposium on Theoretical Aspects of Computer Science (STACS 2016).}}
\titlerunning{The $L_1$ Geodesic Diameter and Center} %optional, in case that the title is too long; the running title should fit into the toppage column

\author[1]{Sang Won Bae}
\author[2]{Matias Korman}
\author[3]{Joseph S.B. Mitchell}
\author[4]{Yoshio Okamoto}
\author[5]{Valentin Polishchuk}
\author[6]{Haitao Wang}

\affil[1]{%Dept. of Computer Science,
Kyonggi University, Suwon, South Korea.
\texttt{swbae@kgu.ac.kr} }

\affil[2]{Tohoku University, Sendai, Japan.
    \texttt{mati@dais.is.tohoku.ac.jp}}

\affil[3]{%Dept. of Computer Science,
Stony Brook University, New York, USA.
\texttt{jsbm@ams.stonybrook.edu}
 }

\affil[4]{%Dept. of Computer Science,
The University of Electro-Communications, Tokyo, Japan.
\texttt{okamotoy@uec.ac.jp}
 }

\affil[5]{%Dept. of Computer Science,
Link\"{o}ping University, Link\"{o}ping, Sweden.
\texttt{valentin.polishchuk@liu.se}
 }

\affil[6]{%Dept. of Computer Science,
Utah State University, Utah, USA.
\texttt{haitao.wang@usu.edu}
 }

\authorrunning{S.W. Bae, M. Korman, J.S.B. Mitchell, Y. Okamoto, V.
Polishchuk, and H. Wang} %mandatory. First: Use abbreviated first/middle names. Second (only in severe cases): Use first author plus 'et. al.'

\Copyright{S.W. Bae, M. Korman, J.S.B. Mitchell, Y. Okamoto, V. Polishchuk, and H. Wang}%mandatory, please use full first names. LIPIcs license is "CC-BY";  http://creativecommons.org/licenses/by/3.0/

% mandatory: Please provide 1-5 keywords
\EventShortName{}
\begin{document}

\maketitle

\begin{abstract}
For a polygonal domain with $h$ holes and a total of $n$ vertices,
we present algorithms that compute the $L_1$ geodesic diameter in $O(n^2+h^4)$ time and
the $L_1$ geodesic center in $O((n^4+n^2 h^4)\alpha(n))$ time, respectively,
where $\alpha(\cdot)$ denotes the inverse Ackermann function.
No algorithms were known for these problems before.
For the Euclidean counterpart, the best algorithms compute
the geodesic diameter in $O(n^{7.73})$ or $O(n^7(h+\log n))$ time, and compute the geodesic center in $O(n^{11}\log n)$ time. Therefore, our algorithms are significantly faster than the algorithms for the Euclidean problems. Our algorithms are based on several interesting observations on $L_1$ shortest paths in polygonal domains.
 \end{abstract}

%%%%%%%%%%%%%%%%%%%%%%%%%%%%%%%%%%%%%%
\section{Introduction}
\label{sec:intro}
%\vspace{-0.1in}
%%%%%%%%%%%%%%%%%%%%%%%%%%%%%%%%%%%%%%%

%Let $\calP$ be a set of $h$ pairwise-disjoint polygonal obstacles with a total of $n$ vertices in the plane. We assume there is a large box that contains all obstacles and consider $\calP$ as the connected region in the box. $\calP$ is also referred to as a \emph{polygonal domain} and the obstacles are {\em holes}.

A \emph{polygonal domain} $\calP$ is a closed and connected polygonal region
in the plane $\Real^2$, with $h\geq 0$ holes (i.e., simple polygons). Let $n$ be the total number of vertices of $\calP$. Regarding the boundary of $\calP$ as obstacles,
we consider shortest obstacle-avoiding paths lying in $\calP$
between any two points $p,q\in\calP$.
Their \emph{geodesic distance} $d(p,q)$ %between any two points $p$ and $q$ in $\calP$
is the length of a shortest path between $p$ and $q$ in $\calP$.
The \emph{geodesic diameter} (or simply {\em diameter}) of $\calP$
is the maximum geodesic distance
over all pairs of points $p,q \in \calP$, i.e., $\max_{p\in\calP} \max_{q\in\calP} d(p,q)$.
Closely related to the diameter is the min-max quantity
$\min_{p\in\calP} \max_{q\in\calP} d(p,q)$,
in which a point $p^*$ that minimizes $\max_{q\in\calP} d(p^*,q)$ is called
a \emph{geodesic center} (or simply {\em center}) of $\calP$.
Each of the above quantities is called \emph{Euclidean} or \emph{$L_1$} depending on
which of the Euclidean or $L_1$ metric is adopted to measure the length of paths.

For simple polygons (i.e., $h=0$), the Euclidean geodesic diameter and center have been studied since the 1980s \cite{ref:AsanoCo85,ref:ChazelleA82,ref:SuriCo89}.
For the diameter, Chazelle~\cite{ref:ChazelleA82} gave the first $O(n^2)$-time algorithm, followed by an $O(n\log n)$-time algorithm by Suri~\cite{ref:SuriCo89}. Finally, Hershberger and Suri~\cite{ref:HershbergerMa97} gave a linear-time algorithm for computing the diameter. For the center, after an $O(n^4\log n)$-time algorithm by Asano and Toussaint~\cite{ref:AsanoCo85}, Pollack, Sharir, and Rote~\cite{ref:PollackCo89} gave an $O(n\log n)$
time algorithm for computing the geodesic center. Recently, Ahn et
al.~\cite{ref:AhnA15} solved the problem in $O(n)$ time.
%It has been a longstanding open problem whether the center can be computed in linear time.

For the general case (i.e., $h>0$), the problems are more difficult.
The Euclidean diameter problem was solved in
$O(n^{7.73})$ or $O(n^7(h+\log n))$ time~\cite{ref:BaeTh13}.
The Euclidean center problem was first solved in $O(n^{12+\epsilon})$ time for any $\epsilon>0$ \cite{ref:BaeCo14CCCG} and then an improved $O(n^{11}\log n)$ time algorithm was given in \cite{ref:WangOn16}.

For the $L_1$ versions, the geodesic diameter and center of simple polygons can be computed in linear time  \cite{ref:BaeCo15,ref:SchuiererCo94}, but we are unaware of any previous algorithms for polygonal domains.
In this paper, we present the first algorithms that compute
the geodesic diameter and center of
a polygonal domain $\calP$ (as defined above)
in $O(n^2 + h^4)$ and $O((n^4+n^2h^4)\alpha(n))$ time,
respectively, where $\alpha(\cdot)$ is the inverse Ackermann function.
Comparing with the algorithms for the same problems under the Euclidean metric, our algorithms are much more efficient, especially when $h$ is significantly smaller than $n$.

%In addition, the diameter and center problems in link metric have also
%been studied extensively,
%e.g.,
%\cite{ref:DjidjevAn92,ref:KeAn89,ref:LenhartCo88,ref:NilssonCo91,ref:NilssonAn96,ref:SuriOn90}.

As discussed in~\cite{ref:BaeTh13}, a main difficulty of polygonal domains seemingly arises
from the fact that there can be several topologically different shortest paths
between two points, which is not the case for simple polygons.
Bae, Korman, and Okamoto~\cite{ref:BaeTh13} observed that
the Euclidean diameter can be realized by two interior points of
a polygonal domain, in which case the two points have at least five distinct shortest paths.
This difficulty makes their algorithm suffer a fairly large running time.
Similar issues also arise in the $L_1$ metric, where a diameter may also be realized by two interior points (this can be seen by extending the examples in~\cite{ref:BaeTh13}).
%Further, under the $L_1$ metric, it seems that at least eight topologically different shortest paths are needed to pin the solution; thus, even if we manage to adapt all techniques used in~\cite{ref:BaeTh13} to the $L_1$ metric, this would result in an algorithm whose running time is significantly larger than $O(n^{7.73})$.
%Note that the algorithm in~\cite{ref:BaeTh13} for the Euclidean diameter
%is not applicable to the $L_1$ case
%due to the linearity of the $L_1$ distance function.

We take a different approach from \cite{ref:BaeTh13}.
We first construct an $O(n^2)$-sized cell decomposition of $\calP$
such that the $L_1$ geodesic distance function restricted to any
pair of two cells can be explicitly described in $O(1)$ complexity.
Consequently, the $L_1$ diameter and center can be obtained by exploring these cell-restricted pieces of the geodesic distance. This leads to simple algorithms that compute the diameter in $O(n^4)$ time and the center in $O(n^6\alpha(n))$ time.
%Moreover, our cell decomposition is based on typical trapezoidal maps
%so that the simplicity of our algorithms is surely retained.
With the help of an ``extended corridor structure'' of $\calP$ \cite{ref:ChenA11ESA,ref:ChenCo12ICALP,ref:ChenL113STACS},
we reduce the $O(n^2)$ complexity of our decomposition to another ``coarser'' decomposition of $O(n+h^2)$ complexity; with another crucial observation (Lemma \ref{lem:bay}), one may compute the diameter in $O(n^3+h^4)$ time by using our techniques for the above $O(n^4)$ time algorithm. One of our main contributions is an additional series of observations (Lemmas \ref{lem:keym} to %\ref{lem:oceanic},lem:c2o-oo,lem:c2o-oceanic,lem:c2o-*x,lem:c2o_vg,lem:c2o-xo,lem:dist_ex,lem:combined_cell,
\ref{lem:dist_vv}) that allow us to further reduce the running time to $O(n^2+h^4)$.
These observations along with the decomposition may have other applications as well.
The idea for computing the center is similar.

We are motivated to study the $L_1$ versions of the diameter and center problems
in polygonal (even non-rectilinear) domains  for several reasons.  First, the $L_1$
metric is natural and well studied in optimization and routing problems, as it models
actual costs in rectilinear road networks and  certain robotics/VLSI applications.
Indeed, the $L_1$ diameter and center problems in the simpler setting of
simply connected domains
have been studied \cite{ref:BaeCo15,ref:SchuiererCo94}.
Second, the $L_1$ metric approximates the Euclidean metric.  Further, improved
understanding of algorithmic results in one metric can assist in understanding in other metrics;
e.g., the continuous Dijkstra methods for $L_1$ shortest paths of \cite{ref:MitchellAn89,ref:MitchellL192} directly led to improved results for Euclidean shortest paths.

\subsection{Preliminaries}
\label{sec:pre}
%%%%%%%%%%%%%%%%%%%%%%%%%%%%%%%%%%%%%%%%%%

For any subset $A\subset \Real^2$, denote by $\bd A$ the
boundary of $A$.  %For any two points $p,q \in \Real^2$,
Denote by $\seg{pq}$ the line segment with endpoints $p$ and $q$.
The {\em $L_1$ length} of $\seg{pq}$ is defined to be $|x_p-x_q|+|y_p-y_q|$, where $x_p$ and $y_p$ are the $x$- and $y$-coordinates of $p$, respectively, and
$x_q$ and $y_q$ are the $x$- and $y$-coordinates of $q$, respectively.
For any polygonal path $\pi \in \Real^2$, let $|\pi|$ be the \emph{$L_1$ length} of $\pi$, which is the sum of the $L_1$ lengths of all segments of $\pi$.
In the following, a path always refers to a polygonal path.
A path is \emph{$xy$-monotone} (or {\em monotone} for short) if every vertical or horizontal line intersects it in at most one connected component.
Following is a basic observation on the $L_1$ length of paths in $\Real^2$,
which will be used in our discussion.

\begin{fact} \label{fact:l1length}
For any monotone path $\pi$ between two points $p,q\in\Real^2$,
$|\pi| = |\seg{pq}|$ holds. %\hfill\qed
\end{fact}

We view the boundary $\bd \calP$ of our polygonal domain $\calP$ as
a series of \emph{obstacles} so that no path in $\calP$ is allowed
to cross $\bd \calP$.
Throughout the paper, unless otherwise stated,
a shortest path always refers to an $L_1$ shortest path and
the distance/length of a path (e.g., $d(p,q)$) always refers to its $L_1$ distance/length.
The diameter/center always refers to the $L_1$ geodesic diameter/center.
For simplicity of discussion, we make a general position assumption
that no two vertices of $\calP$ have the same $x$- or $y$-coordinate.
%On the other hand, when we need to discuss the Euclidean variants of the above concepts,
%we will explicitly attach \emph{Euclidean}, e.g., Euclidean shortest path,
%the Euclidean length of a path, and so on.

The following will also be exploited as a basic fact in further discussion.

\begin{fact}[\cite{ref:GuibasLi87,ref:HershbergerCo94}] \label{fact:euc_simple}
 In any simple polygon $P$, there is a unique Euclidean shortest path $\pi$ between
 any two points in $P$. The path $\pi$ is also an $L_1$ shortest path in
 $P$.
%  \hfill\qed
\end{fact}

The rest of the paper is organized as follows. In Section \ref{sec:td}, we introduce our cell decomposition of $\calP$ and exploit it to have preliminary algorithms for computing the diameter and center of $\calP$. The algorithms will be
improved later in Section \ref{sec:improved}, based on the extended
corridor structure and new observations discussed in Section \ref{sec:corridor}. One may consider the preliminary algorithms in Section \ref{sec:td} relatively straightforward, but we present them for the following reasons. First, they provide an overview on the problem structure. Second, they will help the reader to understand the more sophisticated algorithms given in Section \ref{sec:improved}. Third, some parts of them will also be needed in the algorithms in Section \ref{sec:improved}.

%%%%%%%%%%%%%%%%%%%%%%%%%%%%%%%%%%%%%%%%%%%%%%%%%%%%%%%%
\section{The Cell Decomposition and Preliminary Algorithms}
\label{sec:td}

In this section, we introduce our cell decomposition $\TD$ of $\calP$
and exploit it to have preliminary algorithms that compute the diameter and center of $\calP$.

We first build the \emph{horizontal trapezoidal map} by extending a horizontal line from each vertex of $\calP$ until each end of the line hits $\bd \calP$.
Next, we compute the \emph{vertical trapezoidal map}
by extending a vertical line from each vertex of $\calP$ {\em and} each of the ends of the above extended lines. We then overlay the two trapezoidal maps,
resulting in a {\em cell decomposition} $\TD$ of $\calP$ (e.g., see Fig.~\ref{fig:grid}). The above extended horizontal or vertical line segments are called the \emph{diagonals} of $\TD$.
%In order to distinguish between the vertices of the cells in \td\ and the vertices of $\calP$,
%we call the former the {\em cell vertices} and the latter the
%{\em polygon vertices}.
%A cell vertex may also be a polygon vertex.
Note that $\TD$ has $O(n)$ diagonals and $O(n^2)$ cells.
Each cell $\sigma$ of $\TD$ is bounded by two to four diagonals and at most one edge
of $\calP$, and thus appears as a trapezoid or a triangle;
let $V_\sigma$ be the set of vertices of $\TD$ that are incident to
$\sigma$ (note that $|V_\sigma| \leq 4$).
By an abuse of notation,
we let $\TD$ also denote the set of all the cells of the decomposition.

%Consider a cell $C$ that is a trapezoid. Two of its parallel edges
%must be in diagonals of \td, which are either horizontal
%or vertical, and we call these two edges the {\em bases} of $C$.
%Note that if $C$ is a rectangle, then each of its edges is either
%horizontal or vertical, in which case we pick an arbitrary  pair of parallel edges of
%$C$ as its bases.
%If $C$ is a triangle, then one of its edges must be in a diagonal and
%we consider any such edge $e$ as the base of $C$, and the vertex of $C$
%not incident to $e$ is considered as another ``degenerate'' base of $C$.
%In the following, we view a triangle as a ``degenerate'' trapezoid.

%In the sequel, we partition certain cells of \td\ into different groups and
%our algorithms will treat them differently.
%Let \htd\ be the plane decomposition by removing all vertical
%diagonals from \td. Note that each cell of \htd\ is a trapezoid.
%Similarly, let \vtd\ be the plane decomposition by removing all horizontal
%diagonals from \td, and each cell of \vtd\ is also a trapezoid.
%For each cell $C$ of \td, it is contained in a single cell $C_h$ of \htd\ and
%is also contained in
%a single cell $C_v$ of \vtd. Further, $C$ is exactly the intersection of
%$C_h$ and $C_v$, and each cell vertex of $C$ is the intersection of an
%edge of $C_h$ and an edge of $C_v$.

Each cell of $\TD$ is an intersection between a trapezoid of the horizontal trapezoidal map and another one of the vertical trapezoidal map. Two cells of $\TD$ are \emph{aligned} if they are contained in the same trapezoid of the horizontal
or vertical trapezoidal map, and \emph{unaligned} otherwise.
Lemma \ref{lem:key} is crucial
 for computing both the diameter and the center of $\calP$.

\begin{figure}[tb]
\begin{minipage}[t]{0.49\linewidth}
\begin{center}
\includegraphics[totalheight=1.9in]{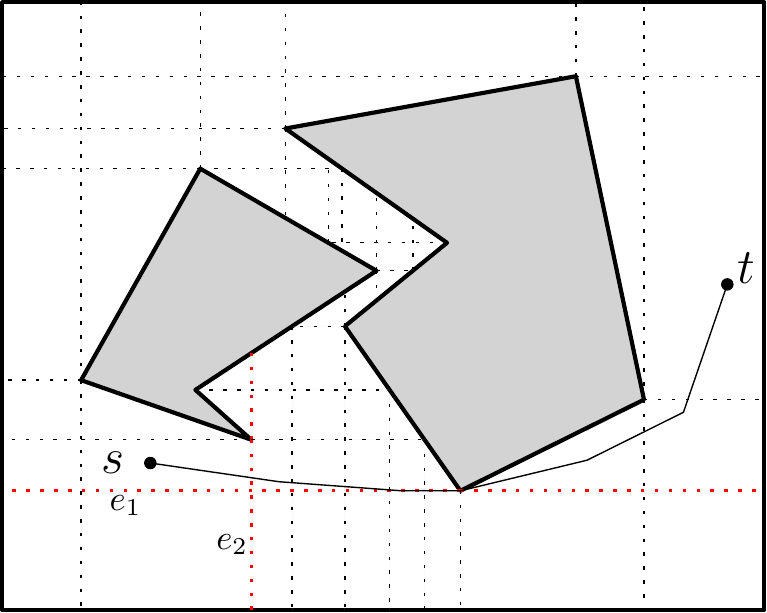}
\caption{\footnotesize The cell decomposition $\TD$ of $\calP$,
and a shortest path from $s$ to $t$.}
\label{fig:grid}
\end{center}
\end{minipage}
\hspace{0.02in}
\begin{minipage}[t]{0.49\linewidth}
\begin{center}
\includegraphics[totalheight=1.9in]{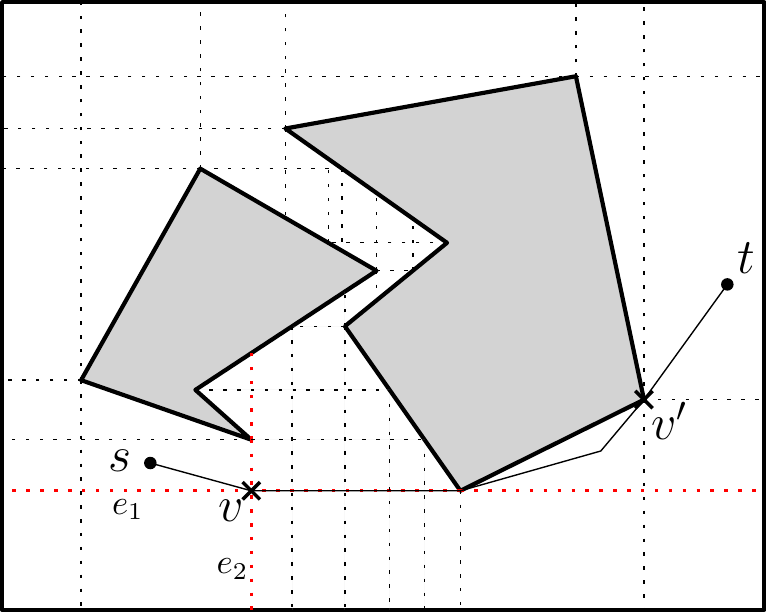}
\caption{\footnotesize Illustrating Lemma~\ref{lem:key}:
a shortest path through vertices $v\in V_\sigma$ and $v'\in V_{\sigma'}$.}
\label{fig:keygrid}
\end{center}
\end{minipage}
\end{figure}

\begin{lemma}\label{lem:key}
Let $\sigma, \sigma'$ be any two cells of $\TD$.
For any point $s\in\sigma$ and any point $t\in\sigma'$,
if $\sigma$ and $\sigma'$ are aligned, then $d(s,t) = |\seg{st}|$;
otherwise, there exists an $L_1$ shortest path between $s$ and $t$
that passes through two vertices $v\in V_\sigma$ and $v'\in V_{\sigma'}$
(e.g., see Fig.~\ref{fig:keygrid}).
\end{lemma}
\begin{proof}
If two cells $\sigma, \sigma'\in \TD$ are aligned,
then they are contained in a trapezoid $\tau$ of the vertical or horizontal trapezoidal map.
Since $\tau$ is convex, any two points in $\tau$ can be joined by a straight segment,
so we have $d(s,t) = |\seg{st}|$ for any $s\in\sigma$ and $t\in \sigma$.

Now, suppose that $\sigma$ and $\sigma'$ are unaligned (e.g., see Fig.~\ref{fig:grid}).
Let $\pi$ be any shortest path between $s$ and $t$.
We first observe that $\pi$ intersects one horizontal diagonal $e_1$ of $\TD$ and
one vertical diagonal $e_2$, both of which bound $\sigma$ (e.g., see Fig.~\ref{fig:grid}, where $e_1$ and $e_2$ are highlighted with red color):
otherwise, $\sigma$ and $\sigma'$ must be aligned.
Since $e_1$ and $e_2$ are bounding $\sigma$,
the intersection $v = e_1 \cap e_2$ is a vertex of $\sigma$ (e.g., see Fig.~\ref{fig:keygrid}).
Let $p_1$ be the first intersection of $\pi$ with $e_1$ while we go along $\pi$ from $s$ to $t$.
Similarly, define $p_2$ to be the first intersection of $\pi$ with $e_2$.

Since $e_1$ is horizontal and $e_2$ is vertical, the union of the two line segments
$\seg{p_1v}\cup\seg{vp_2}$ is a shortest path between $p_1$ and $p_2$.
We replace the portion of $\pi$ between $p_1$ and $p_2$
by $\seg{p_1v}\cup\seg{vp_2}$ to have another $s$-$t$ path $\pi'$.
Since $\seg{p_1v}\cup\seg{vp_2}$ is monotone,
its length is equal to $|\seg{p_1p_2}|$ by Fact~\ref{fact:l1length}.
This implies that $\pi'$ is a shortest path between $s$ and $t$
and passes through the vertex $v$.

Symmetrically, the above argument can be applied to the other side,
the destination $t$ and the cell $\sigma'$,
which implies that $\pi'$ can be modified to a \st\ shortest path $\pi''$
that passes through $v$ and simultaneously a vertex $v'$ of $\sigma'$.
The lemma thus follows.
\end{proof}

\subsection{Computing the Geodesic Diameter}

In this section, we present an $O(n^4)$ time algorithm for computing
the diameter of $\calP$.

%Now, we are ready to describe our first algorithms.
%and the general idea will also be used in the
%improved $O(n^2+h^4)$-time algorithm presented in Section~\ref{sec:improved}.

The general idea is to consider every pair of cells of $\TD$ separately.
For each pair of such cells $\sigma, \sigma' \in \TD$,
we compute the maximum geodesic distance between $\sigma$ and $\sigma'$,
that is, $\max_{s\in\sigma, t\in\sigma'} d(s,t)$, called
the \emph{$(\sigma, \sigma')$-constrained diameter}.
Since $\TD$ is a decomposition of $\calP$,
the diameter of $\calP$ is equal to the maximum value of the constrained diameters
over all pairs of cells of $\TD$.
We handle two cases depending on whether $\sigma$ and $\sigma'$ are aligned.

If $\sigma$ and $\sigma'$ are aligned, by Lemma~\ref{lem:key}, for any $s\in \sigma$ and $t\in\sigma'$,
we have $d(s,t) = |\seg{st}|$, i.e, the $L_1$ distance of $\seg{st}$.
Since the $L_1$ distance function is convex,
the $(\sigma, \sigma')$-constrained diameter is always realized by
some pair $(v, v')$ of two vertices with $v\in V_\sigma$ and $v'\in V_\sigma'$.
We are thus done by checking at most $16$ pairs of vertices, in $O(1)$ time.

In the following, we assume that $\sigma$ and $\sigma'$ are unaligned.
%This case also makes use of Lemma~\ref{lem:key}.
Consider any point $s\in\sigma$ and any point $t\in\sigma'$.
For any vertex $v\in V_\sigma$ and any vertex $v'\in V_{\sigma'}$,
consider the path from $s$ to $t$ obtained by
concatenating $\seg{sv}$, a shortest path from $v$ to $v'$, and $\seg{v't}$, and
let $d_{vv'}(s,t)$ be its length.
Lemma~\ref{lem:key} ensures that $d(s,t) = \min_{v\in V_\sigma, v'\in V_{\sigma'}} d_{vv'}(s,t)$.
Since $d_{vv'}(s,t) = |\seg{sv}| + |\seg{v't}| + d(v, v')$ and
$d(v,v')$ is constant over all $(s,t)\in \sigma \times \sigma'$,
the function $d_{vv'}$ is linear on $\sigma \times \sigma'$.
Thus, it is easy to compute the $(\sigma, \sigma')$-constrained diameter
once we know the value of $d(v, v')$ for every pair $(v,v')$ of vertices.

\begin{lemma}\label{lem:constrained_diam}
For any two cells $\sigma, \sigma' \in \TD$,
the $(\sigma, \sigma')$-constrained diameter can be
computed in constant time,
provided that $d(v, v')$ for every pair $(v,v')$ with
$v\in V_\sigma$ and $v'\in V_{\sigma'}$ has been computed.
\end{lemma}
\begin{proof}
The case where $\sigma$ and $\sigma'$ are aligned is easy
as discussed above.
We thus assume they are unaligned.

Assume that we know the value of $d(v, v')$
for every pair $(v, v')$ with $v\in V_\sigma$ and $v'\in V_{\sigma'}$.
Recall that $d(s,t) = \min_{v\in V_\sigma, v'\in V_{\sigma}} d_{vv'}(s,t)$
and $d_{vv'}(s,t) =  |\seg{sv}| + |\seg{v't}| + d(v, v')$.
Further, note that $|V_\sigma|\leq 4, |V_{\sigma'}| \leq 4$.

Since each $d_{vv'}$ is a linear function on its domain $\sigma \times \sigma'$,
its graph appears as a hyperplane in a $5$-dimensional space.
Thus, the geodesic distance function $d$ restricted on $\sigma \times \sigma'$
corresponds to the lower envelope of those hyperplanes.
Since there are only a constant number of pairs $(v, v')$,
the function $d$ can also be explicitly constructed in $O(1)$ time.
Finally, we find the highest point on the graph of $d$
by traversing all of its faces.
\end{proof}

For each vertex $v$ of $\TD$, a straightforward method can compute $d(v, v')$
for all other vertices $v'$ of $\TD$ in $O(n^2\log n)$ time, by first
computing the shortest path map $\spm(v)$
\cite{ref:MitchellAn89,ref:MitchellL192} in $O(n\log n)$ time and then
computing $d(v, v')$ for all $v'\in \TD$ in $O(n^2\log n)$ time.
We instead give a faster sweeping algorithm in Lemma \ref{lem:findcells} by making use of the
property that all vertices on every diagonal of $\TD$ are sorted.

\begin{lemma}\label{lem:findcells}
For each vertex $v$ of $\TD$,
we can evaluate $d(v, v')$ for all vertices $v'$ of $\TD$ in $O(n^2)$ time.
\end{lemma}
\begin{proof}
Our algorithm attains its efficiency by using the property that all
vertices on each diagonal of $\TD$ are sorted.
Specifically, suppose that $\TD$ is represented by a standard data
structure, e.g., the doubly connected edge list. %\cite{ref:deBergCo08}.
Then, by traversing each diagonal (either vertical or horizontal), we
can obtain a vertically or horizontally sorted list of vertices on
that diagonal.

We first compute the shortest path map $\spm(v)$ in $O(n \log n)$ time~\cite{ref:MitchellAn89,ref:MitchellL192}.
We then apply a standard sweeping technique, say, we sweep $\spm(v)$
by a vertical line from left to right.
The events are when the sweep line hits vertices of $\spm(v)$, obstacle vertices of $\calP$, or
vertical diagonals of $\TD$.
Note that each vertex of $\TD$ is
either on a vertical diagonal or an obstacle vertex.
We use the standard technique to handle the events of vertices of $\spm(v)$ and $\calP$,
and each such event costs $O(\log n)$ time.
For each event of a vertical diagonal of $\TD$, we simply do a linear search on the
sweeping status to find the cells of $\spm(v)$ that contain the cell vertices of
$\TD$ on the diagonal.
Each such event takes $O(n)$ time since each diagonal of $\TD$ has $O(n)$ vertices.
Note that the total number of events is $O(n)$.
Hence, the running time of the sweeping algorithm is $O(n^2)$.
\end{proof}

Thus, after $O(n^4)$-time preprocessing, for any two cells $\sigma, \sigma' \in \TD$,
the $(\sigma, \sigma')$-constrained diameter can be computed in $O(1)$ time
by Lemma~\ref{lem:constrained_diam}.  Since $\TD$ has $O(n^2)$ cells,
it suffices to handle at most $O(n^4)$ pairs of cells,
resulting in $O(n^4)$ candidates for the diameter, and the maximum is the diameter.
%is determined simply by taking their maximum.
Hence, we obtain the following result.

\begin{theorem} \label{thm:td_diam}
The $L_1$ geodesic diameter of $\calP$
can be computed in $O(n^4)$ time. %\hfill\qed
\end{theorem}

\subsection{Computing the Geodesic Center}
\label{sec:center}

We now present an algorithm that computes an $L_1$ center of $\calP$.
The observation in Lemma~\ref{lem:key} plays an important role in our algorithm.

%In the following, a center of $\calP$ always refers to an $L_1$ geodesic
%center of $\calP$.

For any point $q\in\calP$, we define $\dmax(q)$ to be the maximum geodesic distance
between $q$ and any point in $\calP$, i.e., $\dmax(q) := \max_{p\in\calP} d(p,q)$.
A center $q^*$ of $\calP$ is defined to be a point
with $\dmax(q^*) = \min_{q\in\calP} \dmax(q)$.
Our approach is again based on the decomposition $\TD$:
for each cell $\sigma \in \TD$, we want to find a point $q \in\sigma$
that minimizes the maximum geodesic distance $d(p, q)$ over all $p\in \calP$.
We call such a point $q \in \sigma$ a \emph{$\sigma$-constrained center}.
Thus, if $q'$ is a $\sigma$-constrained center,
then we have $\dmax(q') = \min_{q\in\sigma} \dmax(q)$.
%Since $\min_{q\in\calP} \dmax(q) = \min_{\sigma\in\TD} \min_{q\in\sigma} \dmax(q)$,
Clearly, the center $q^*$ of $\calP$ must be a $\sigma$-constrained center
for some $\sigma \in\TD$.
Our algorithm thus finds a $\sigma$-constrained center for every $\sigma\in\TD$,
which at last results in $O(n^2)$ candidates for a center of $\calP$.

Consider any cell $\sigma\in\TD$. To compute
a $\sigma$-constrained center, we investigate the function $R$ restricted to $\sigma$
and exploit Lemma~\ref{lem:key} again. To utilize Lemma~\ref{lem:key},
for any point $q\in \sigma$,
we define $\dmax_{\sigma'}(q):=\max_{p\in\sigma'} d(p,q)$ for any $\sigma' \in \TD$.
For any $q\in \sigma$, $\dmax(q) = \max_{\sigma'\in\TD} \dmax_{\sigma'}(q)$, that is,
$R$ is the upper envelope of all the $\dmax_{\sigma'}$ on the domain $\sigma$.
Our algorithm explicitly computes the functions $\dmax_{\sigma'}$ for all $\sigma' \in \TD$ and computes the upper envelope $\calU$ of the graphs of the $\dmax_{\sigma'}$.
Then, a $\sigma$-constrained center corresponds to a lowest point on
$\calU$.

We observe the following for the function $\dmax_{\sigma'}$.

\begin{lemma} \label{lem:R}
The function $\dmax_{\sigma'}$ is piecewise linear on $\sigma$ and has
$O(1)$ complexity.
%Moreover, there exists a set $C\subset \Real$ with $|C| = O(1)$
%such that for any $\sigma'\in\TD$, any linear patch of $\dmax_{\sigma'}$
%is described by a linear equation with all coefficients in $C$.
\end{lemma}
\begin{proof}
Recall that $\dmax_{\sigma'}(q) = \max_{p\in \sigma'} d(p, q)$ for $q\in \sigma$.
In this proof, we regard $d$ to be restricted on $\sigma \times \sigma' \subset \Real^4$
and use a coordinate system of $\Real^4$ by introducing $4$ axes, $x$, $y$, $u$, and $w$
with $p=(x,y)\in \sigma$ and $q=(u,w)\in \sigma'$.
Thus, we may write $\dmax_{\sigma'}(q) = \dmax_{\sigma'}(x,y) = \max_{(u,w)\in\sigma'} d(x,y,u,w)$.

The graph of function $d$ consists of $O(1)$ linear patches as shown
in the proof of Lemma~\ref{lem:constrained_diam}.
Once we fully identify the geodesic distance function $d$ on $\sigma\times\sigma'$,
we consider its graph $S := \{z = d(x,y,u,w)\}$
for all $(x,y,u,w)\in\sigma\times\sigma'$, which is a hypersurface in
a $5$-dimensional space $(\sigma\times\sigma')\times\Real \subset \Real^5$
with an additional axis $z$.
We then project the graph $S$ onto the $(x,y,z)$-space.
More precisely, the projection of $S$ is the set
$\{(x,y,z) \mid (x,y,u,w,z) \in S\}$.
Thus, for any $(x,y)\in \sigma$, $\dmax_{\sigma'}(x,y)$ is determined by
the highest point in the intersection of the projection with the $z$-axis parallel line
through point $(x,y,0)$.
This implies that the function $\dmax_{\sigma'}$ simply corresponds to
the upper envelope (in the $z$-coordinate) of the projection of $S$.
Since $S$ consists of $O(1)$ linear patches,
so does the upper envelope of its projection,
which concludes the proof.
\end{proof}

Now, we are ready to describe how to compute a $\sigma$-constrained center.
We first handle every cell $\sigma'\in\TD$ to compute the graph of $\dmax_{\sigma'}$
and thus gather its linear patches.
Let $\Gamma$ be the family of those linear patches for all $\sigma' \in \TD$.
We then compute the upper envelope of $\Gamma$ and find a lowest point
on the upper envelope, which corresponds to a $\sigma$-constrained center.
Since $|\Gamma| = O(n^2)$ by Lemma~\ref{lem:R},
the upper envelope can be computed in $O(n^4\alpha(n))$ time
by executing the algorithm by Edelsbrunner et al.~\cite{ref:EdelsbrunnerTh89},
where $\alpha(\cdot)$ denotes the inverse Ackermann function.
The following theorem summarizes our algorithm.

\begin{theorem}\label{theo:center}
An $L_1$ geodesic center of $\calP$ can be computed in $O(n^6\alpha(n))$ time.
\end{theorem}
\begin{proof}
As a preprocessing, we compute all of the vertex-to-vertex geodesic distances
$d(v, v')$ for all pairs of vertices of $\TD$ in $O(n^4)$ time.
We show that for a fixed $\sigma \in \TD$,
a $\sigma$-constrained center can be computed in $O(n^4\alpha(n))$ time.
As discussed in the proof of Lemma~\ref{lem:constrained_diam},
for any $\sigma'\in\TD$,
the geodesic distance function $d$ restricted on $\sigma\times \sigma'$,
along with its graph $D$ over $\sigma \times\sigma'$,
can be specified in $O(1)$ time.
By Lemma~\ref{lem:R},
from $D$ we can describe the function $\dmax_{\sigma'}$ in $O(1)$ time.
The last task is to compute the upper envelope of all $\dmax_{\sigma'}$
in $O(n^4\alpha(n))$ time, as discussed above,
by executing the algorithm by Edelsbrunner et al.~\cite{ref:EdelsbrunnerTh89}.
\end{proof}

%%%%%%%%%%%%%%%%%%%%%%%%%%%%%%%%%%%%%%%%%%%%%%%%%%%%%%%%%
\section{Exploiting the Extended Corridor Structure}
\label{sec:corridor}

In this section, we briefly review the extended corridor structure of $\calP$
and present new observations, which will be crucial for
our improved algorithms in Section~\ref{sec:improved}.
The corridor structure has been used for solving shortest path
problems~\cite{ref:ChenA11ESA,ref:InkuluPl09,ref:KapoorAn97}.
Later some new concepts such as ``bays,'' ``canals,''
and the ``ocean'' were introduced~\cite{ref:ChenCo12ICALP,ref:ChenL113STACS},
referred to as the ``extended corridor structure.''

\subsection{The Extended Corridor Structure}
\label{sec:corridornew}

Let $\Tri$ denote an arbitrary triangulation of $\calP$ (e.g.,
see \figurename~\ref{fig:triangulation}).
We can obtain $\Tri$ in $O(n\log n)$ time or $O(n+h\log^{1+\epsilon} h)$ time
for any $\epsilon>0$~\cite{ref:Bar-YehudaTr94}.
Let $G$ denote the dual graph of $\Tri$,
i.e., each node of $G(\calP)$ corresponds to a triangle in
$\Tri(\calP)$ and each edge connects two nodes of $G(\calP)$
corresponding to two triangles sharing a diagonal of $\Tri(\calP)$.
%The degree of each node in $G(\calP)$ is at most three.  As in
%\cite{ref:KapoorAn97}, at least one node dual to a triangle incident
%to each of $s$ and $t$ is of degree three.
%Based on $G$, we compute a planar $3$-regular graph,
%denoted by $G^3$ (the degree of each node in $G^3$ is three),
%possibly with loops and multi-edges, as follows.
%First, we remove every degree-one node from $G$ along with its incident edge;
%repeat this process until no degree-one node exists.
%Second, remove every degree-two node from $G$ and replace its two incident edges
%by a single edge; repeat this process until no degree-two node exists.
Based on $G$, one can obtain a planar $3$-regular graph,
possibly with loops and multi-edges,
by repeatedly removing all degree-one nodes and then contracting all degree-two nodes.
The resulting $3$-regular graph has $O(h)$ faces, nodes, and edges~\cite{ref:KapoorAn97}.
Each node of the graph corresponds to a triangle in $\Tri$, called a \emph{junction triangle}.
The removal of all junction triangles from $\calP$ results in $O(h)$ components,
called \emph{corridors}, each of which corresponds to an edge of the
graph.  See \figurename~\ref{fig:triangulation}.
Refer to \cite{ref:KapoorAn97} for more details.

\begin{figure}[t]
\begin{minipage}[t]{0.46\linewidth}
\begin{center}
\includegraphics[totalheight=1.5in]{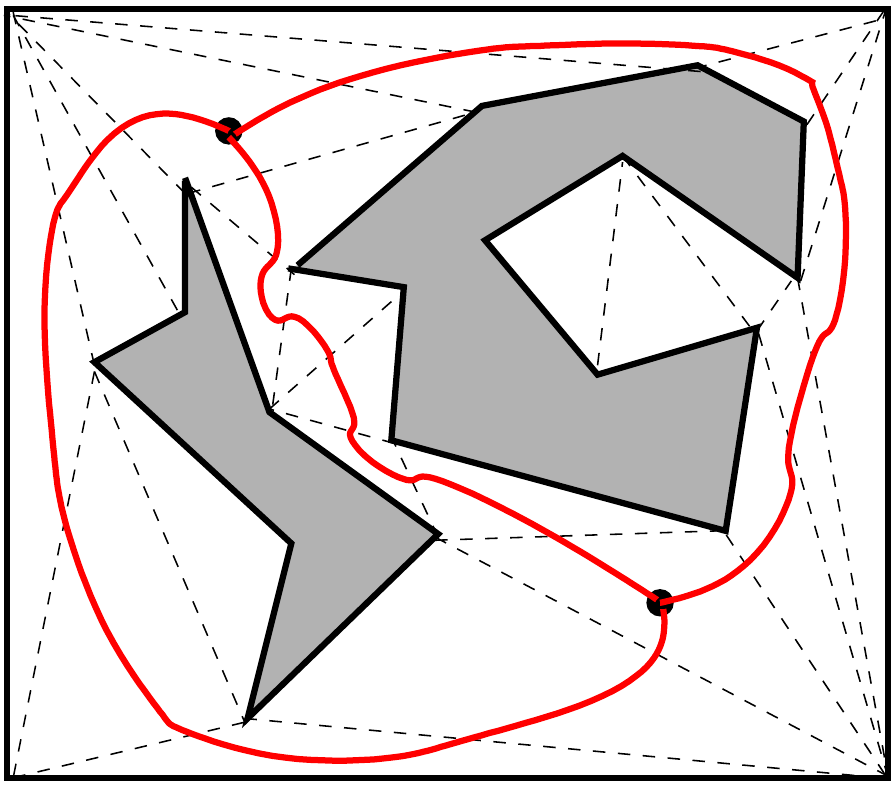}
\caption{\footnotesize A triangulation $\Tri$ of $\calP$
and the $3$-regular graph obtained from the dual graph of $\Tri$
whose nodes and edges are depicted by black dots and red solid curves.
Each junction triangle corresponds to a node and
removing all junction triangles results in three corridors, in this figure,
each of which corresponds to an edge of the graph.
}
\label{fig:triangulation}
\end{center}
\end{minipage}
\hspace*{0.02in}
\begin{minipage}[t]{0.53\linewidth}
\begin{center}
\includegraphics[totalheight=1.5in]{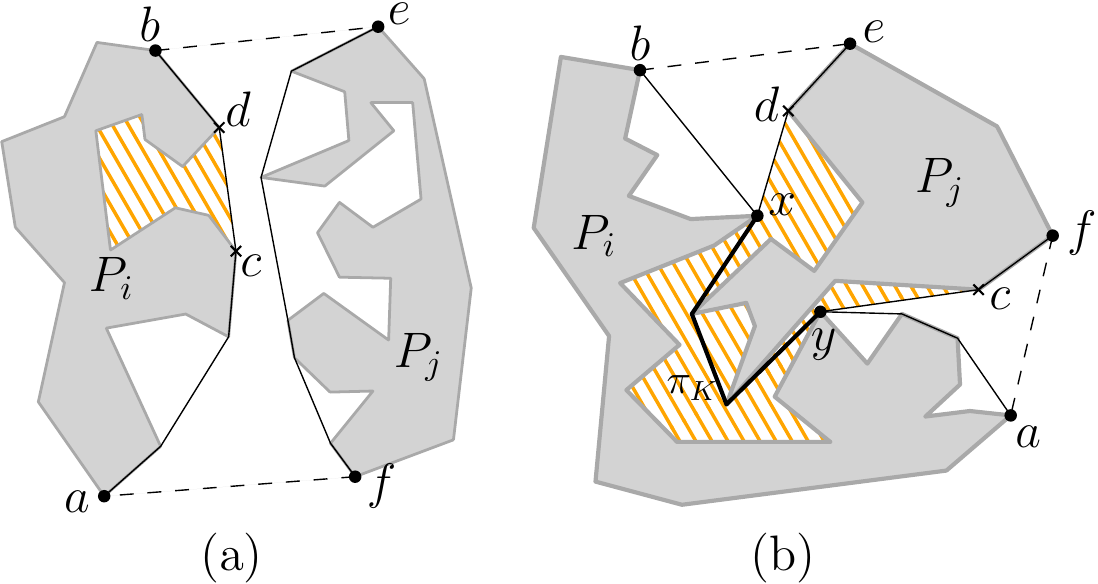}
\caption{\footnotesize Hourglasses $H_K$ in corridors $K$.
The dashed segments $\seg{be}$ and $\seg{af}$ are diagonals of junction triangles in $\Tri$.
(a) $H_K$ is open.
Five bays can be seen. A bay with gate $\seg{cd}$ is shown as the shaded region.
(b) $H_K$ is closed. There are three bays and a canal, and the shaded
region depicts the canal with two gates
$\seg{dx}$ and $\seg{cy}$.
}
\label{fig:corridor}
\end{center}
\end{minipage}
\end{figure}

Next we briefly review the concepts of bays, canals, and the ocean.
Refer to~\cite{ref:ChenCo12ICALP,ref:ChenL113STACS} for more details.

Let $P_1, \ldots, P_h$ be the $h$ holes of $\calP$ and $P_0$ be the outer polygon of $\calP$.
For simplicity, a hole may also refer to the unbounded region outside $P_0$ hereafter.
The boundary $\bd K$ of a corridor $K$ consists of
two diagonals of $\Tri$ and two paths along the boundary of holes $P_i$ and $P_j$,
respectively (it is possible that $P_i$ and $P_j$ are the
same hole, in which case one may consider $P_i$ and $P_j$ as the above
two paths respectively).
Let $a, b\in P_i$ and $e, f\in P_j$ be the endpoints of the two paths,
respectively, such that $\seg{be}$ and $\seg{fa}$ are diagonals of $\Tri$,
each of which bounds a junction triangle.  See \figurename~\ref{fig:corridor}.
Let $\pi_{ab}$ (resp., $\pi_{ef}$) denote the \emph{Euclidean} shortest path from $a$ to $b$
(resp., $e$ to $f$) inside $K$.
The region $H_{K}$ bounded by $\pi_{ab}, \pi_{ef}$, $\seg{be}$, and $\seg{fa}$
is called an \emph{hourglass}, which is either \emph{open}
if $\pi_{ab}\cap \pi_{ef}=\emptyset$, or \emph{closed}, otherwise.
If $H_{K}$ is open, then both $\pi_{ab}$ and $\pi_{ef}$ are convex chains
and are called the \emph{sides} of $H_{K}$;
otherwise, $H_{K}$ consists of two ``funnels'' and a path
$\pi_{K}=\pi_{ab}\cap \pi_{ef}$ joining the two apices of the two funnels,
called the \emph{corridor path} of $K$.
The two funnel apices (e.g., $x$ and $y$ in \figurename~\ref{fig:corridor}(b))
connected by $\pi_{K}$ are called the \emph{corridor path terminals}.
Note that each funnel comprises two convex chains.

%We compute the hourglass for each corridor. After the triangulation, computing the
%hourglasses for all corridors takes totally $O(n)$ time.

We consider the region of $K$ minus the interior of $H_K$, which consists of a number of
simple polygons facing (i.e., sharing an edge with) one or both of $P_i$ and $P_j$.
We call each of these simple polygons a \emph{bay} if it is facing a single hole,
or a \emph{canal} if facing both holes.
Each bay is bounded by a portion of the boundary of a hole and
a segment $\seg{cd}$ between two obstacle vertices
$c, d$ that are consecutive along a side of $H_K$.
We call the segment $\seg{cd}$ the \emph{gate} of the bay.
(See \figurename~\ref{fig:corridor}(a).)
On the other hand, there exists a unique canal for each corridor $K$
only when $H_K$ is closed and the two holes $P_i$ and $P_j$ both bound the canal.
The canal in $K$ in this case completely contains the corridor path $\pi_K$.
A canal has two \emph{gates} $\seg{xd}$ and $\seg{yc}$
that are two segments facing the two funnels, respectively,
where $x, y$ are the corridor path terminals and $d, c$ are vertices of the funnels.
(See \figurename~\ref{fig:corridor}(b).)
Note that each bay or canal is a simple polygon.

%If $H_R$ is open, then the region $H\setminus H_R$ consists of a number of connected components,
%each of which is called a \emph{bay}.
%Each bay is bounded by a portion of the boundary of a hole and a segment $\seg{cd}$ between two obstacle vertices
%$c, d$ that are consecutive along a side of $H_R$.
%See the left of \figurename~\ref{fig:corridor}.
%We call the segment $\seg{cd}$ the \emph{gate} of the bay.
%% denoted by $bay(\overline{cd})$.
%
%If $H_R$ is closed, let $x$ and $y$ be the two apices of its two funnels.
%Again, the region $H\setminus H_R$ consists of bays (that are bounded as described above)
%and one connected component $A$ that completely contains the corridor path between $x$ and $y$.
%We call the region $A$ the \emph{canal} of $H_R$.
%A canal $A$ has two \emph{gates} $\seg{xd}$ and $\seg{yz}$
%that are two segments facing the two funnels, respectively,
%where $d$ and $z$ are vertices of the funnels.
%See the right of \figurename~\ref{fig:corridor}.
%%The canal containing the corridor between $x$ and $y$ is denoted by $canal(x,y)$.
%Note that each bay or canal is a simple polygon.

Let $\calM \subseteq \calP$ be the union of all junction triangles,
open hourglasses, and funnels.
We call $\calM$ the \emph{ocean}.
Its boundary $\bd \calM$ consists of $O(h)$ convex vertices and $O(h)$ reflex chains
each of which is a side of an open hourglass or of a funnel.
Note that $\calP \setminus \calM$ consists of all bays and canals of $\calP$.

For convenience of discussion, define each bay/canal in such a way that they do not
contain their gates and hence their gates are contained in $\calM$;
therefore, each point of $\calP$ is either in a bay/canal or in $\calM$, but not in both. After the triangulation $\Tri$ is obtained, computing the ocean, all bays and canals can be done in $O(n)$ time \cite{ref:ChenCo12ICALP,ref:ChenL113STACS,ref:KapoorAn97}.

Roughly speaking, the reason we partition $\calP$ into the ocean, bays, and canals is to facilitate evaluating the distance $d(s,t)$ for any two points $s$ and $t$ in $\calP$. For example, if both $s$ and $t$ are in $\calM$, then we can use a similar method as in Section~\ref{sec:td} to evaluate $d(s,t)$. However, the challenging case happens when one of $s$ and $t$ is in $\calM$ and the other is in a bay or canal.

The following lemma is one of our key observations for our improved algorithms in Section~\ref{sec:improved}. It essentially tells that for any point $s\in \calP$ and any bay or canal $A$, the farthest point of $s$ in $A$ is achieved on the boundary $\partial A$, which is similar in spirit to the simple polygon case.

\begin{lemma}\label{lem:bay}
Let $s\in\calP$ be any point and $A$ be a bay or canal of $\calP$.
Then, for any $t\in A$,
there exists $t'\in \bd A$ such that $d(s,t) \leq d(s,t')$.
Equivalently, $\max_{t\in A} d(s,t) = \max_{t\in \bd A} d(s,t)$.
\end{lemma}
%\section*{Proof of Lemma~\ref{lem:bay}}
\begin{proof}
Recall that the gates of $A$ are not contained in $A$ but in the ocean $\calM$.
Let $\bar{A}$ be the closure of $A$, that is, $\bar{A}$ consists of $A$ and its gates.
For any $p,q\in \bar{A}$, let $d_A(p, q)$ be the $L_1$ geodesic
distance in $\bar{A}$.
Since $\bar{A}$ is a simple polygon,
Fact~\ref{fact:euc_simple} implies that there is a unique Euclidean shortest path
$\pi_2(p, q)$ in $\bar{A}$ between any $p,q\in \bar{A}$,
and $d_A(p,q) = |\pi_2(p,q)|$.
In general, we have $d_A(p, q) \geq d(p,q)$.

Depending on whether $A$ is a bay or a canal, our proof will consider two cases. We first prove a basic property as follows.

%\subsubsection*{A Basic Property}
\paragraph*{A Basic Property}
Consider any point $s'\in \calP$ and any point $t'\in A$.
Then, we claim that there exists a shortest $s'$-$t'$ path $\pi$ with the following property (*):
\begin{quote}
 \textit{%
 (*) $\pi$ crosses each gate of $A$ at most once and each component of $\pi\cap \bar{A}$
 is the unique Euclidean shortest path $\pi_2(p, q)$ for some points $p,q\in\bar{A}$.}
\end{quote}
Consider any shortest $s'$-$t'$ path $\pi'$.
If $\pi'$ crosses a gate $g$ of $A$ at least twice,
then let $p\in g$ and $q \in g$ be the first and last points on $g$
we encounter when walking along $\pi'$ from $s'$ to $t'$.
We can replace the portion of $\pi'$ between $p$ and $q$ with the line
segment $\seg{pq}$ by Fact~\ref{fact:l1length}
to obtain another shortest path that crosses $g$ at most once.
One can repeat this procedure for all gates of $A$ to have
a shortest path $\pi''$ crossing each gate of $A$ at most once.
Then, we take a connected component of $\pi'' \cap \bar{A}$,
which is an $L_1$ shortest path between its two endpoints $p, q$ inside $\bar{A}$.
This implies that $d(p, q) = d_A(p,q) = |\pi_2(p,q)|$,
so we can replace the component by $\pi_2(p,q)$.
Repeat this for all components of $\pi'' \cap \bar{A}$ to obtain
another shortest $s'$-$t'$ path $\pi$ with the desired property (*).

\begin{figure}[t]
\begin{minipage}[t]{\linewidth}
\begin{center}
\includegraphics[width=0.98\textwidth]{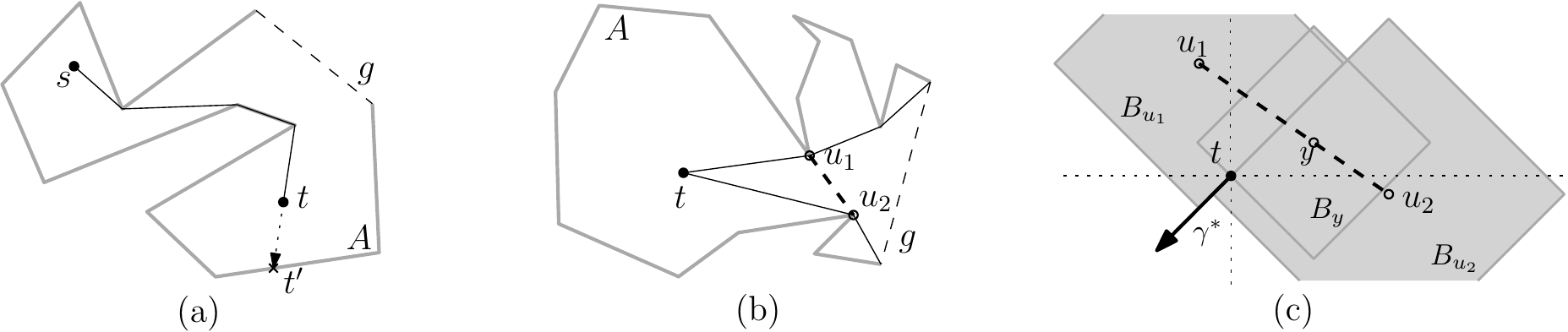}
\caption{\footnotesize Illustration to the proof of Lemma~\ref{lem:bay} when $A$ is a bay.
(a) When $s\in \bar{A}$ and (b)(c) when $s\notin \bar{A}$.
}
\label{fig:bay}
\end{center}
\end{minipage}
\end{figure}

%\subsubsection*{The Bay Case}
\paragraph*{The Bay Case}
To prove the lemma, we first prove the case where $A$ is a bay.
Then, $A$ has a unique gate $g$.
Recall that the gate $g$ is not contained in $A$.
%Let $\bar{A} = A \cup g$ be the closure of $A$, that is, $\bar{A}$ is a closed, simple polygon.
Depending on whether $s$ is in $\bar{A}$, there are two cases.
\begin{itemize}
\item
Suppose that $s\in \bar{A}$.
Let $t \in A$ be any point in $A$.
By our claim, there exists a shortest \st\ path $\pi$ in $\calP$ with property (*).
Since both $s$ and $t$ lie in $\bar{A}$, $\pi$ does not cross $g$, and
is thus contained in $\bar{A}$.
Moreover, by property (*), we have $\pi = \pi_2(s, t)$.
Hence, $d(s,t) = d_A(s,t)$.
%This implies that the farthest point $t'$ over $\bar{A}$ from $s$
%also maximizes the geodesic distance in $\bar{A}$.
%We then apply the known fact by Schuierer~\cite{ref:SchuiererCo94}
%that

If $t$ lies on $\bd A$, then the lemma trivially holds.
Suppose that $t$ lies in the interior of $A$.
Then, we can extend the last segment of $\pi_2(s,t)$
until it hits a point $t'$ on the boundary $\bd A$.
See \figurename~\ref{fig:bay}(a).
Again since $\bar{A}$ is a simple polygon,
the extended path is indeed $\pi_2(s, t')$.
By the above argument, $d(s,t') = d_A(s, t') = |\pi_2(s,t')|$,
which is strictly larger than $|\pi_2(s,t)| = d(s,t)$.

%\begin{figure}[t]
%\begin{minipage}[t]{0.49\linewidth}
%\begin{center}
%\includegraphics[totalheight=1.2in]{crossbay.eps}
%\caption{\footnotesize Illustrating a global shortest \st\ path that
%crosses $\overline{cd}$ at two points $p$ and $q$. }
%\label{fig:crossbay}
%\end{center}
%\end{minipage}
%\hspace*{0.02in}
%\begin{minipage}[t]{0.49\linewidth}
%\begin{center}
%\includegraphics[totalheight=1.2in]{baycone.eps}
%\caption{\footnotesize The arrow shows the opposite direction of the
%bisector of the angle formed by $\overline{sc'}$ and $\overline{sd'}$
%and facing $\overline{cd}$.}
%\label{fig:baycone}
%\end{center}
%\end{minipage}
%\end{figure}

\item
Suppose that $s \notin \bar{A}$.
Then, any shortest path from $s$ to any point $t\in A$
must cross the gate $g$.
This implies that $d(s,t) = \min_{x\in g}\{d(s, x) + d_A(x, t)\}$ for any $t\in A$.
We show that there exists $t'\in \bd A$ such that
for any point $x\in g$, it holds that $d_A(x, t) \leq d_A(x,t')$,
which implies that $d(s,t) = \min_{x\in g}\{d(s, x) + d_A(x, t)\}
\leq \min_{x\in g}\{d(s, x) + d_A(x, t')\} = d(s,t')$.
For the purpose, we consider the union of $\pi_2(x, t)$
for all $x\in g$.
The union forms a funnel $\funnel_g(t)$ plus the Euclidean shortest path $\pi_2(u, t)$ from
the apex $u$ of $\funnel_g(t)$ to $t$.
If $u\neq t$, then we extend the last segment of $\pi_2(u, t)$ to a point $t'$ on
the boundary $\bd A$, similarly to the previous case so that
$d_A(u, t) \leq d_A(u, t')$ and thus $d_A(x, t) \leq d_A(x, t')$ for any $x\in g$.
Otherwise, if $u = t$, then
let $u_1$ and $u_2$ be the two vertices of $\funnel_g(t)$ adjacent to the apex $t$.
See \figurename~\ref{fig:bay}(b).
%Let $C$ be the cone formed by $u_1$ and $u_2$ with apex $t$ facing to $g$.

Observe that the segment $\seg{u_1u_2}$ separates $t$ and the gate $g$,
and hence path $\pi_2(x, t)$ for any $x\in g$ crosses $\seg{u_1u_2}$.
We now claim that there exists a ray $\gamma^*$ from $t$ such that
as $t$ moves along $\gamma^*$, $d_A(t, y)$ for any fixed $y\in \seg{u_1u_2}$
is nondecreasing.
If the claim is true, then we select $t' = \gamma^* \cap \bd A$, so the lemma follows
since it holds that $d_A(t, x) = \min_{y_\in\seg{u_1u_2}} \{d_A(t, y) + d_A(y, x)\}$
$\leq \min_{y_\in\seg{u_1u_2}} \{d_A(t', y) + d_A(y, x)\} = d_A(t',
x)$ for any $x\in g$. Next, we prove the claim.

For each $y\in \seg{u_1u_2}$, let $B_y$ be the $L_1$ disk centered at $y$ with radius
$d_A(y,t) = |\seg{yt}|$.
See \figurename~\ref{fig:bay}(c).
Since $t$ lies on the boundary of $B_y$, as we move $t$ along a ray $\gamma$ in some direction
outwards $B_y$, $|\seg{yt}|$ is not decreasing, that is,
for any $p\in \gamma$, $|\seg{yp}| \geq |\seg{yt}|$.
This also implies that $d_A(y, p) \geq |\seg{yp}| \geq |\seg{yt}| = d_A(y, t)$.
Let $h_y$ be the set of all such rays $\gamma$
that as $t$ moves along $\gamma$, $|\seg{yt}|$ is not decreasing.
Our goal is thus to show that $\bigcap_{y\in\seg{u_1u_2}} h_y \neq \emptyset$,
and pick $\gamma^*$ as any ray in the intersection.
For the purpose, we consider the four quadrants---left-upper, right-upper,
left-lower, and right-lower---centered at $t$.
Then, since the $B_y$ are all $L_1$ disks,
the set $h_y$ only depends on which quadrant $y$ belongs to;
more precisely, for any $y$ in a common quadrant, the set $h_y$ stays constant.
For example, for any $y \in \seg{u_1u_2}$ lying in the right-upper quadrant,
then $h_y$ is commonly the set of all rays from $t$
in direction between $135^\circ$ and $315^\circ$, inclusively,
since $t$ lies on the bottom-left edge of $B_y$ in this case.
Thus, the directions of all rays in $h_y$ span an angle of exactly $180^\circ$.
Moreover, $\seg{u_1u_2}$ is a line segment and thus intersects only three quadrants.
Therefore, $\bigcap_{y\in\seg{u_1u_2}} h_y$ is equal to the intersection
of at most three different sets of rays, whose directions span an angle of $180^\circ$.
(\figurename~\ref{fig:bay}(c) illustrates an example scene when $\seg{u_1u_2}$
intersects three quadrants centered at $t$.)
This implies that $\bigcap_{y\in\seg{u_1u_2}}h_y \neq \emptyset$.
\end{itemize}

\begin{figure}[t]
\begin{minipage}[t]{\linewidth}
\begin{center}
\includegraphics[totalheight=1.7in]{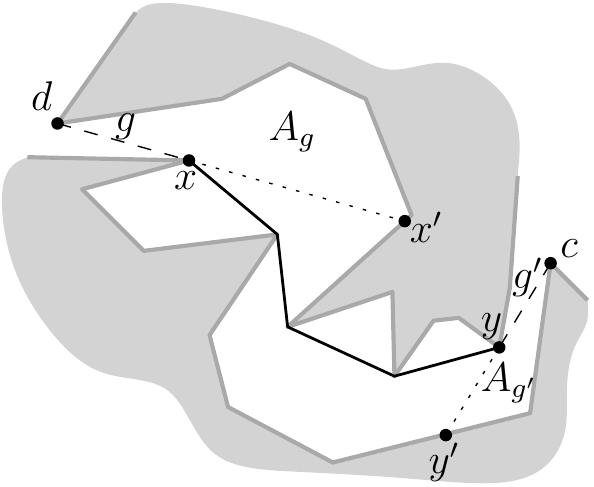}
\caption{\footnotesize Illustration to the proof of Lemma~\ref{lem:bay} when $A$ is a canal
with two gates $g=\seg{dx}$ and $g'=\seg{cy}$.
}
\label{fig:obsercanal}
\end{center}
\end{minipage}
\end{figure}

%\subsubsection*{The Canal Case}
\paragraph*{The Canal Case}
Above, we have proved the lemma for the bay case where $A$ is a bay.
Next, we turn to the canal case: suppose that $A$ is a canal.

Let $g=\seg{dx}$ and $g'=\seg{cy}$ be the two gates of $A$,
where $x$ and $y$ are the two corridor path terminals
(e.g., see \figurename~\ref{fig:obsercanal}).
We extend $g$ from $x$ into the interior of $A$, in the direction opposite to $d$.
Note that due to the definition of canals, this extension always goes into
the interior of $A$ (refer to~\cite{ref:ChenCo12ICALP} for detailed discussion).
Let $x'$ be the first point on $\bd A$ hit by the extension.
The line segment $\seg{xx'}$ partitions $A$ into
two simple polygons,
and the one containing $d$ is denoted by $A_g$.
We consider $\seg{dx'}$ as an edge of $A_g$,
but for convenience of discussion, we assume that $A_g$ does not contain
the segment $\seg{dx'}$.
Define $A_{g'}$ analogously for the other gate $g' = \seg{cy}$.
If $t\in A_g$, then we can view $A_g$ as a ``bay'' with gate $\seg{dx'}$,
and apply the identical argument as done in the bay case,
concluding that for any $t\in A_g$ there exists $t'\in \bd A_g$ such that
$d(s,t) \leq d(s,t')$.
If $t' \in \bd A$, then we are done.
Otherwise, if $t'\in \seg{xx'}$, then $s \in A_g$ according to our analysis on the bay case.
In this case, we move $t'$ along $\seg{xx'}$ to $x$ or $x'$,
since $\max\{ d(s, x), d(s, x')\} \geq d(s, t') \geq d(s,t)$, we are done.
The case of $t\in A_{g'}$ is analogous.

Let $\widehat{A} := A \setminus (A_g\cup A_{g'})$.
From now on, we suppose $t \in \widehat{A}$.
Observe that
for any $p\in g$, $\pi_2(p, t)$ passes through the corridor path
terminal $x$ since $t\notin A_g$ \cite{ref:ChenCo12ICALP};
symmetrically, for any $p'\in g'$, $\pi_2(p', t)$ passes through $y$.
Consider any $L_1$ shortest \st\ path $\pi$ in $\calP$ with property (*).
We classify $\pi$ into one of the following three types:
(a) $\pi$ lies inside $A$, that is, $\pi = \pi_2(s, t)$,
(b) when walking along $\pi$ from $s$ to $t$, the last gate crossed by $\pi$ is $g$, or
(c) is $g'$.
Note that $\pi$ falls into one of the three cases.
In case (a), indeed we have $s\in A$ and $d(s, t) = d_A(s, t)$.
In case (b), $\pi$ consists of a shortest path from $s$ to $x$
and $\pi_2(x, t)$, and thus $d(s,t) = d(s,x) + d_A(x, t)$.
Symmetrically, in case (c), we have $d(s,t) = d(s,y) + d_A(y,t)$.

Depending on whether $s\in \bar{A}$ or not, we handle two possibilities.
In the following, we assume property (*) when we discuss any shortest \st\ path.
\begin{itemize}
\item Suppose that $s\notin \bar{A}$.
Then, any shortest \st\ path in $\calP$ must cross a gate of $A$.
This means that there is no shortest \st\ path of type (a),
and we have $d(s,t) = \min\{d(s,x) + d_A(x, t), d(s, y) + d_A(y,t)\}$.
Consider a decomposition of $\widehat{A}$ into three regions $R_x$,
$R_y$, and $B$
such that $R_x = \{ p \in \widehat{A} \mid d(s,x) + d_A(x,p) < d(s,y) + d_A(y,p)\}$,
$R_y = \{ p \in \widehat{A} \mid d(s,x) + d_A(x,p) > d(s,y) + d_A(y,p)\}$,
and $B =  \{ p \in \widehat{A} \mid d(s,x) + d_A(x,p) = d(s,y) + d_A(y,p)\}$.
This decomposition is clearly the geodesic Voronoi diagram in simple polygon $\widehat{A}$
of two sites $\{x, y\}$ with additive weights. See Aronov~\cite{ref:AronovOn89} and Papadopoulou and Lee~\cite{ref:PapadopoulouNe89}.
Also, we have that $d(s,t) = d(s,x) + d_A(x, t)$ for any $t\in R_x$;
$d(s,t) = d(s,y) + d_A(y,t)$ for any $t\in R_y$.
The region $B$ is called the bisector between $x$ and $y$.
By the property of Voronoi diagrams~\cite{ref:AronovOn89,ref:PapadopoulouNe89},
$B = \bd R_x \cap \bd R_y$
and $B$ is a path connecting two points on $\bd \widehat{A}$.
Let $b_0 \in B$ be the intersection $\pi_2(x, y)\cap B$.
Then, $d_A(x, b_0) = d_A(y, b_0) \leq d_A(x, b) = d_A(y, b)$ for any $b\in B$,
and moreover if we move $b$ along $B$ in one direction from $b_0$,
$d_A(x, b)$ is nondecreasing.
Thus, $\max_{b\in B} d_A(x,b)$ is attained when $b$ is an endpoint of $B$.

When $t\in \bd A$, the lemma is trivial.
If $t\in B$, then we let $t'$ be the endpoint of $B$ in direction away from $b_0$.
Then, by the property of the bisector $B$, we have
$d_A(x, t) \leq d_A(x, t')$ and $d_A(y,t) \leq d_A(y, t')$,
and hence $d(s, t) \leq d(s, t')$.
If $t'$ lies on the boundary $\bd A$, we conclude the lemma;
otherwise, $t'$ may lie on $\bd \widehat{A}\setminus \bd A$, say on $\seg{xx'}$.
In this case, we apply the analysis of the bay case where $s\notin A_g$
to find a point $t''$ on $\bd A$ such that $d(s, t') \leq d(s, t'')$.

If $t$ lies in the interior of $R_x$,
then we extend the last segment of $\pi_2(x, t)$ until it hits
a point $t'$ on $\bd R_x$.
Then, we have that $d(s, t') \geq d(s,t)$.
Note that $t'$ lies on $\bd \widehat{A}$ or on $B$;
in any case, we apply the above argument so that we can find a point $t''$ on $\bd A$
with $d(s, t'') \geq d(s, t') \geq d(s,t)$.
The case where $t$ lies in the interior of $R_y$ is handled analogously.

\item Finally, suppose that $s\in \bar{A}$.
%In this case, any shortest \st\ path in $\calP$ of type (b) or (c)
%must cross both gates $g$ and $g'$.
%Here, we claim that we can ignore one of types (b) or (c)
%by showing that any path of type (b) or (c) is always dominated by
%the path of type (a), $\pi_2(s,t)$.
We again consider the $L_1$ geodesic Voronoi diagram in $\widehat{A}$
of three sites $\{s, x, y\}$ with additive weights $0, d(s,x), d(s,y)$, respectively.
As done above, we observe that the maximum value $\max_{t\in \widehat{A}} d(s, t)$
is attained when $t\in \bd A$ or $t$ is a Voronoi vertex.
The former case is analyzed above.
Here, we prove that the latter case cannot happen.
Note that there are three shortest paths of different types between $s$ and the Voronoi vertex
while there are exactly two shortest paths of types (b) and (c) to any point
on the bisector between $x$ and $y$.
In the following, we show that the bisector between $x$ and $y$ cannot appear
in the Voronoi diagram, which implies that the Voronoi diagram has no vertex.

\begin{figure}[t]
\begin{minipage}[t]{0.46\linewidth}
\begin{center}
\includegraphics[totalheight=1.0in]{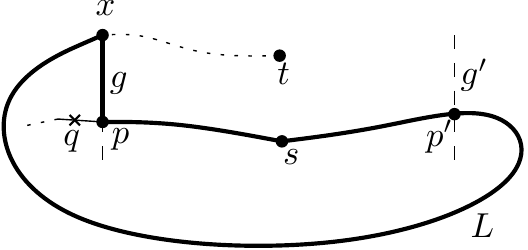}
\caption{\footnotesize Illustration to the proof of Lemma~\ref{lem:bay} when $A$ is a canal.}
\label{fig:loopL}
\end{center}
\end{minipage}
\hspace*{0.02in}
\begin{minipage}[t]{0.53\linewidth}
\begin{center}
\includegraphics[totalheight=1.5in]{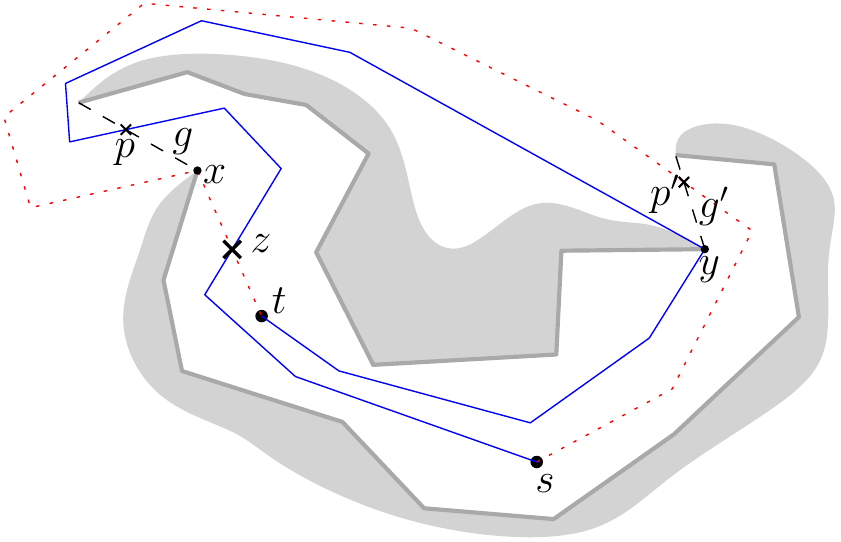}
\caption{\footnotesize Illustrating two shortest \st\ paths $\pi$ (red dotted) and $\pi'$ (blue  solid), intersecting at $z$, in a canal with two gates $g$ and $g'$.}
\label{fig:pathreplace}
\end{center}
\end{minipage}
\end{figure}

%Since $A_g$ and $A_{g'}$ are disjoint,
%either $s\notin A_g$ or $s\notin A_{g'}$.
%Without loss of generality, we assume that $s\notin A_{g}$.
Suppose to the contrary that the bisector between $x$ and $y$ appears as a nonempty Voronoi edge
of the Voronoi diagram, and that $t$ is a point on it.
That is, $d(s, t) = d(s, x) + d_A(x, t) = d(s, y) + d_A(y, t) > d_A(s, t)$.
Let $\pi$ and $\pi'$ be two shortest \st\ paths of type (b) and (c), respectively.
Thus, $\pi$ passes through $x$ and $\pi'$ passes through $y$ to reach $t\in\widehat{A}$.
By property (*), $\pi$ crosses $g'$ first and then $g$ when walking along $\pi$ from $s$ to $t$,
and $\pi'$ crosses $g$ first and then $g'$.
Let $p\in g$ and $p'\in g'$ be the last point of $\pi' \cap g$ and $\pi \cap g'$,
respectively, when we walk along $\pi$ and $\pi'$ from $s$ to $t$.

We claim that $\pi$ and $\pi'$ intersect each other in a point other than $s$ and $t$.
%Assume to the contrary that $\pi\cap\pi'=\{s, t\}$.
Indeed, consider the loop $L$ formed by the subpath of $\pi$ between $s$ and $x$,
the segment $\seg{xp}$, and the subpath of $\pi'$ between $s$ and $p$.
See \figurename~\ref{fig:loopL}.
Also, let $U$ be a disk centered at $p$ with arbitrarily small radius
and $q \in \pi' \cap \bd U$ be the point on $\pi'$ not lying in $A$.
If the loop $L$ does not separate $q$ and $t$, then the subpath of $\pi$ from $x$ to $t$ must intersect $\pi'$ in a point other than $t$ (see Fig.~\ref{fig:pathreplace}), and thus the claim follows; otherwise, the subpath of $\pi'$ from $q$ to $t$ must cross $L$ at some point other than $s$ and $t$, and thus the claim also follows.

Let $z\in \pi\cap\pi' \setminus\{s, t\}$ (see Fig.~\ref{fig:pathreplace}), and $\pi_{zt}$ and $\pi'_{zt}$ be
the subpath of $\pi$ and $\pi'$, respectively, from $z$ to $t$.
By the property of shortest paths, we have $|\pi_{zt}| = |\pi'_{zt}| = d(z,t)$.
Hence, replacing $\pi_{zt}$ by $\pi'_{zt}$ in $\pi$ results in
another shortest \st\ path $\pi''$.
If $\pi''$ lies inside $\bar{A}$, then we have $d(s,t) = |\pi''| \geq d_A(s,t)$.
%Otherwise, $\pi''$ crosses $g$ twice at $p$ and $x$
%and thus replacing the subpath of $\pi''$ from $p$ to $x$ by $\seg{px}$
%results in another shortest path lying inside $\bar{A}$.
Otherwise, $\pi''$ crosses $g'$ twice at $p'$ and $y$,
and thus replacing the subpath of $\pi''$ from $p'$ to $y$
by $\seg{p'y}$ results in another shortest path inside $\bar{A}$.
In either way, there is another $L_1$ shortest \st\ path of type (a), and
hence $d(s,t) = d_A(s,t)$,
a contradiction to the assumption that $t$ lies on the bisector between $x$ and $y$.
%
%
%By the above observation, since $s\notin A_{g}$,
%the subpath of $\pi'$ from $s$ to a point on $g$
%passes through $x$; that is, $\pi'$ consists of $\pi_2(s, x)$ and a shortest path
%from $x$ to $t$ that is of type (b).
%Since $\pi$ is also a shortest \st\ path,
%we know that $d(x, t) = d_A(x,t)$.
%We thus obtain another shortest \st\ path $\pi''$
%inside $A$ by replacing the subpath of $\pi'$ by $\pi_2(x,t)$.
%This implies that $d_A(s,t) \leq |\pi''| \leq |\pi'|$, and thus $d_A(s,t) = d(s,t)$.
%Hence, we have $d_A(s, t) \leq |\pi'| = d(s, y) + d_A(y, t)$ for any $t\in \widehat{A}$,
%and thus we can ignore the paths of type (c).
%
%In this case, we consider the geodesic Voronoi diagram in $\widehat{A}$
%of sites $\{s, x\}$ with additive weights $d(s,s)=0$ and $d(s,x)$ to define three regions
%$R_s := \{ t\in \widehat{A} \mid d_A(s,t) < d(s, x) + d_A(x, t)\}$,
%$R_x := \{ t\in \widehat{A} \mid d_A(s,t) > d(s, x) + d_A(x, t)\}$,
%and $B:= \{ t\in \widehat{A} \mid d_A(s,t) = d(s, x) + d_A(x, t)\}$.
%Then, we conclude the lemma by the analogous discussion as in the above case
%where $s\notin A$.
%The other case where the paths of type (b) can be ignored
%can be handled in a symmetric fashion
%by considering the Voronoi diagram of sites $\{s, y\}$.
\end{itemize}

This finishes the proof of the lemma.
\end{proof}

\subsection{Shortest Paths in the Ocean $\calM$}

We now discuss shortest paths in the ocean $\calM$.
Recall that corridor paths are contained in canals, but their terminals are on $\bd\calM$.
By using the corridor paths and $\calM$,
finding an $L_1$ or Euclidean shortest path between two points $s$ and $t$ in $\calM$
can be reduced to the convex case since $\bd \calM$ consists of $O(h)$ convex chains.
For example, suppose both $s$ and $t$ are in $\calM$.
Then, there must be a shortest $s$-$t$ path $\pi$ that lies in the
union of $\calM$ and all corridor paths~\cite{ref:ChenA11ESA,ref:ChenL113STACS,ref:KapoorAn97}.

%Our algorithm will use some techniques given in
%\cite{ref:ChenA11ESA,ref:ChenCo12arXiv} and we first briefly discuss
%them below (refer to \cite{ref:ChenA11ESA,ref:ChenCo12arXiv} for more
%details).
%Recall that $\partial\calM$ consists of $O(h)$ convex chains of $O(n)$ vertices, and
%there are $O(h)$ corridor paths and their terminals are on
%$\partial\calM$.
Consider any two points $s$ and $t$ in $\calM$. A shortest \st\ path $\pi(s,t)$ in
$\calP$ is a shortest path in $\calM$ that possibly contains some
corridor paths. Intuitively, one may view corridor paths as
``shortcuts'' among the components of the space $\calM$.
As in \cite{ref:KapoorAn97}, since
$\bd \calM$ consists of $O(h)$ convex vertices and $O(h)$ reflex chains,
the complementary region $\calP' \setminus\calM$ (where $\calP'$ refers to the union of $\calP$ and all its holes) can be partitioned into a set $\calB$ of $O(h)$ convex objects with a total of $O(n)$ vertices (e.g., by extending an angle-bisecting segment inward from each convex vertex \cite{ref:KapoorAn97}).
If we view the objects in $\mathcal{B}$ as obstacles,
then $\pi$ is a shortest path avoiding all
obstacles of $\calB$ but possibly containing some corridor paths.
Note that our algorithms can work on $\calP$ and $\calM$ directly without
%partitioning $\calP\setminus\calM$ into
using $\calB$; but for ease of exposition, we will discuss our algorithm with the help of $\calB$.

%\begin{figure}[t]
%\begin{minipage}[t]{\linewidth}
%\begin{center}
%\includegraphics[totalheight=1.0in]{core.eps}
%\caption{\footnotesize Illustrating the core of a convex obstacle (suppose the red points are corridor path terminals).}
%\label{fig:core}
%\end{center}
%\end{minipage}
%\vspace*{-0.2in}
%\end{figure}

Each convex obstacle $Q$ of $\calB$ has at most four {\em extreme vertices}: the topmost, bottommost, leftmost, and rightmost vertices, and
there may be some corridor path terminals on the boundary of $Q$.
We connect the extreme vertices and the corridor path
terminals on $\partial Q$ consecutively by line segments to obtain
another polygon, denoted by $\c(Q)$ and called the {\em core} of $Q$ (see \figurename~\ref{fig:core}).
%Denote by $\c(\calB)$ the set of all $O(h)$
%cores of the obstacles in $\calB$.
Let $\calP_{core}$ denote the complement of the union of all cores $\c(Q)$ for all $Q\in\calB$ and corridor paths in $\calP$.
Note that the number of vertices of $\calP_{core}$ is $O(h)$
and $\calM\subseteq \calP_{core}$.
For $s,t\in\calP_{core}$, let $d_{core}(s,t)$ be the geodesic distance
between $s$ and $t$ in $\calP_{core}$.

\begin{figure}[t]
\begin{minipage}[t]{0.45\linewidth}
\begin{center}
\includegraphics[totalheight=1.3in]{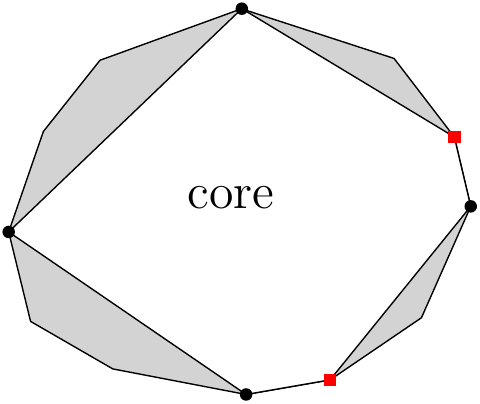}
\caption{\footnotesize Illustrating the core of a convex obstacle: the
(red) squared vertices are corridor path terminals. }
\label{fig:core}
\end{center}
\end{minipage}
\hspace*{0.02in}
\begin{minipage}[t]{0.54\linewidth}
\begin{center}
\includegraphics[totalheight=1.8in]{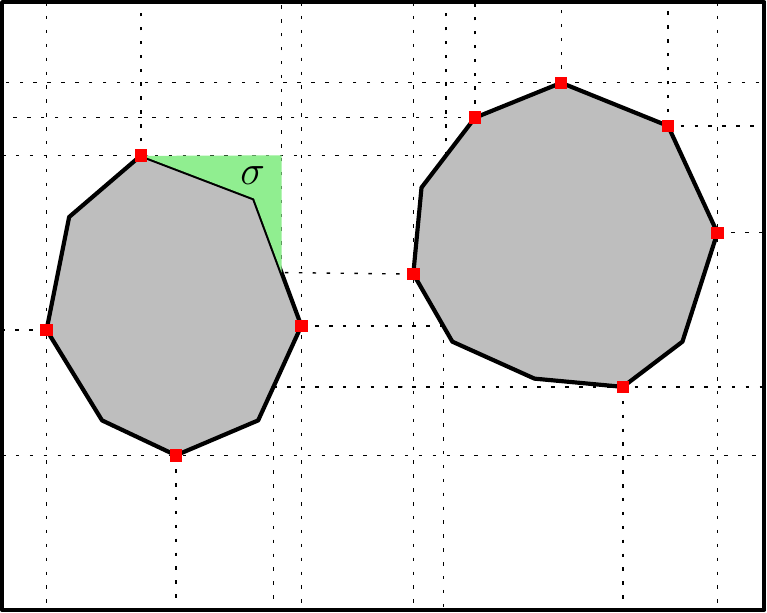}
\caption{\footnotesize Illustrating the core-based cell decomposition
$\TDM$: the (red) squared vertices are core vertices and the green cell $\sigma$ is a boundary cell.}
\label{fig:oceangrid}
\end{center}
\end{minipage}
\end{figure}

The core structure leads to a more efficient way
to find an $L_1$ shortest path between two points in $\calP$.
Chen and Wang~\cite{ref:ChenA11ESA}
proved that an $L_1$ shortest path between $s,t\in\calM$ in $\calP_{core}$
can be locally modified to an $L_1$ shortest path in $\calP$
without increasing its $L_1$ length.

\begin{lemma}[\cite{ref:ChenA11ESA}]\label{lem:core}
For any two points $s$ and  $t$ in $\calM$, $d(s,t) = d_{core}(s,t)$ holds.
\end{lemma}

%Recall that $\pi(s,t)$ is a shortest \st\ path.
%It was proved in \cite{ref:ChenA11ESA,ref:ChenCo12arXiv} that a
%shortest path from $s$ to $t$ that avoids all cores and possibly
%contains some corridor paths has the same $L_1$ length as $\pi(s,t)$
%(i.e., $d(s,t)$).
%The proof is based on some observations, which we briefly discuss below
%because our algorithms also need them.
%
%Consider any edge $\overline{uv}$ of $\c(P)$. The edge $\overline{uv}$ divides
%$P$ into two polygons and we call the one that does not contain
%$\c(P)$ the {\em ear} of $\overline{uv}$, denoted by
%$ear(\overline{uv})$ (e.g., see Fig.~\ref{fig:core}).
%Denote by $\partial(\overline{uv})$ the boundary of
%$ear(\overline{uv})$ except $\overline{uv}$, i.e.,
%$\partial(\overline{uv})$ is the intersection of $ear(\overline{uv})$
%and $\partial\calP$.
%
%
%\begin{lemma}{\em\cite{ref:ChenA11ESA}}\label{lem:core}
%For any line segment $\overline{w_1w_2}$ such that both $w_1$ and
%$w_2$ are on $\partial(\overline{uv})$, the $L_1$ length of
%$\overline{w_1w_2}$ is the same as that of the portion of
%$\partial(\overline{uv})$ between $w_1$ and $w_2$.
%\end{lemma}

Hence, to compute $d(s,t)$ between two points $s$ and $t$ in $\calM$,
it is sufficient to consider only the cores and the corridor paths, that is, $\calP_{core}$.
We thus reduce the problem size from $O(n)$ to $O(h)$.
Let $\spm_{core}(s)$ be a shortest path map for any source point $s\in \calM$.
Then, $\spm_{core}(s)$ has $O(h)$ complexity and
can be computed in $O(h \log h)$ time~\cite{ref:ChenA11ESA}.
%by slightly modifying Mitchell's continuous Dijkstra algorithm~\cite{ref:MitchellL192}.

\subsection{Decomposition of the Ocean $\calM$}

We introduce a core-based cell decomposition $\TDM$ of the ocean $\calM$  (see \figurename~\ref{fig:oceangrid})
in order to fully exploit the advantage of the core structure
in designing algorithms computing the $L_1$ geodesic diameter and center.
For any $Q\in\calB$, the vertices of $\c(Q)$ are called
\emph{core vertices}.

The construction of $\TDM$ is analogous to that of
the previous cell decomposition $\TD$ for $\calP$.
We first extend a horizontal line
only from each \emph{core vertex} until it hits $\bd\calM$ to have a horizontal diagonal,
and then extend a vertical line from each core vertex and each endpoint of
the above horizontal diagonal.
The resulting cell decomposition induced by the above diagonals is $\TDM$. Hence, $\TDM$ is constructed in $\calM$ with respect to core vertices.
Note that $\TDM$ consists of $O(h^2)$ cells and
can be built in $O(n\log n + h^2)$ time by a typical plane sweep algorithm.
We call a cell $\sigma$ of $\TDM$
a \emph{boundary cell} if $\bd \sigma \cap \bd \calM \neq \emptyset$.
For any boundary cell $\sigma$,
the portion $\bd\sigma\cap\bd\calM$ appears as a convex chain of $Q\in\calB$
by our construction of its core and $\TD_\calM$; since
$\bd\sigma\cap\bd\calM$ may contain multiple vertices of $\calM$, the complexity of $\sigma$ may not be constant.
Any non-boundary cell of $\TDM$ is a rectangle bounded by four diagonals.
Each vertex of $\TDM$ is either an endpoint of its diagonal or an intersection
of two diagonals; thus, the number of vertices of $\TDM$ is $O(h^2)$.

Below we prove an analogue of Lemma~\ref{lem:key} for the decomposition $\TDM$ of $\calM$.
Let $V_\sigma$ be the set of vertices of $\TDM$ incident to $\sigma$.
Note that $|V_\sigma| \leq 4$.
We define the alignedness relation between two cells of $\TDM$
analogously to that for $\TD$.
We then observe an analogy to Lemma \ref{lem:key}.

\begin{lemma}  \label{lem:keym}
Let $\sigma, \sigma'$ be any two cells of $\TDM$.
If they are aligned, then $d(s,t) = |\seg{st}|$ for any $s\in\sigma$ and $t\in\sigma'$;
otherwise, there exists a shortest \st\ path in $\calP$ containing
two vertices $v\in V_\sigma$ and $v'\in V_{\sigma'}$ with
$d(s,t) = |\seg{sv}| + d(v,v') + |\seg{v't}|$.
\end{lemma}
\begin{proof}
We first discuss the case where $\sigma$ and $\sigma'$ are aligned.
In this case, they are bounded by two consecutive parallel diagonals of $\TDM$,
and let $S\subset \calM$ be the region in between the two diagonals.
Since $S$ consists of two monotone concave chains and the two diagonals
by our construction of $\TDM$,
it is not difficult to see that any $s\in \sigma$ and $t\in\sigma'$
can be joined by a monotone path $\pi$ inside $S$.
This implies that $|\pi| = |\seg{st}| = d(s,t)$ by Fact~\ref{fact:l1length}.

Next, we consider the unaligned case.
Suppose that $\sigma$ and $\sigma'$ are unaligned.
By Lemma~\ref{lem:core},
there exists a shortest \st\ path $\pi$ in $\calP$ such that
$\pi$ lies inside the union $\calM$ and all the corridor paths of $\calP$.
%
%Let $\pi_{core}$ be any shortest path between $s$ and $t$ in $\calP_{core}$.
%We know that $|\pi_{core}| = d(s,t)$ by Lemma~\ref{lem:core}.
%Here, we construct a specific shortest \st\ path $\pi$ in $\calP$
%by performing the following procedure on $\pi_{core}$:
%%By applying Lemma~\ref{lem:core}, we modify $\pi_{core}$ to result in
%%another path $\pi$ with $|\pi| = |\pi_{core}|$ by performing the following procedure:
%for every subpath $\pi_{pq}$ of $\pi_{core}$ between two points $p$ and $q$
%that is a connected component of the intersection of $\pi_{core}$ and a hole $P_i$ of $\calP$,
%we replace $\pi$ by the chain $\pi'_{pq}$ along the boundary of $P_i$ from $p$ to $q$.
%Observe that the chain $\pi'_{pq} \subset \bd P_i$ contains no core vertex of $\calP_{core}$;
%otherwise, $\pi_{core}$ is not a feasible path inside $\calP_{core}$.
%This implies that the chain $\pi'_{pq}$ is convex and monotone
%by our construction of $\calP_{core}$.
%We thus have $|\pi'_{pq}| = |\seg{pq}|$ by Fact~\ref{fact:l1length}
%and the replacements in the procedure do not increase the length of the resulting path $\pi$,
%that is, $|\pi| = |\pi_{core}| = d(s,t)$.
%It is also obvious that $\pi$ lies completely inside $\calP$.
%Hence, $\pi$ is a shortest path between $s$ and $t$ in $\calP$.
%Note that by our construction
%if $\pi$ intersects a canal $A$, then it must follow the corridor path in $A$
%via its two terminals; no bays intersect $\pi$ on the other hand.
Our proof for this case is analogous to that of Lemma~\ref{lem:key}.
Since $\sigma$ and $\sigma'$ are unaligned,
there are two possibilities when we walk along $\pi$ from $s$ to $t$:
either we meet a horizontal diagonal $e_1$ and a vertical diagonal $e_2$ of $\TDM$
that bound $\sigma$, or enter a corridor path via its terminal $x$.
In the former case, we can apply the same argument as done in
the proof of Lemma~\ref{lem:key} to show that $\pi$ can be modified
to pass through a vertex $v\in V_\sigma$ with $v=e_1 \cap e_2$ without increasing
the length of the resulting path.
In the latter case, observe by our construction of $\TDM$ that
$x$ is also a vertex of $\TDM$ and there is a diagonal extended from $x$.
If $x\in V_\sigma$, we are done since $d(s,x) = |\seg{sx}|$ as discussed above
(any cell $\sigma$ is aligned with itself).
Otherwise, there is a unique cell $\sigma'' \neq \sigma \in\TDM$ with $x\in V_{\sigma''}$
that is aligned with $\sigma$, and there is a common diagonal $e$ bounding
$\sigma$ and $\sigma''$.
In this case, since $\pi$ passes through $x$, it indeed intersects two diagonals,
which means that this is the former case.
\end{proof}

%%%%%%%%%%%%%%%%%%%%%%%%%%%%%%%%%%%%%%%%%%%%%%%%%%%%
\section{Improved Algorithms}
\label{sec:improved}
%%%%%%%%%%%%%%%%%%%%%%%%%%%%%%%%%%%%%%%%%%%%%%%%%%%%

In this section, we further explore the geometric
structures and give more observations about our decomposition. These results, together with our results in Section~\ref{sec:corridor}, help us to give
improved algorithms that compute
the diameter and center, using a similar
algorithmic framework as in Section~\ref{sec:td}.

%%%%%%%
\subsection{The Cell-to-Cell Geodesic Distance Functions}

Recall that our preliminary algorithms in Section~\ref{sec:td}
rely on the nice behavior of the cell-to-cell
geodesic distance function: specifically, $d$ restricted to $\sigma\times\sigma'$
for any two cells $\sigma, \sigma'\in\TD$ is the lower envelope of $O(1)$ linear functions.
We now have two different cell decompositions, $\TD$ of $\calP$ and $\TDM$ of $\calM$.
Here, we observe analogues of Lemmas~\ref{lem:key} and \ref{lem:keym}
for any two cells in $\TD \cup \TDM$,
by extending the alignedness relation between cells in $\TD$ and $\TDM$, as follows.

Consider the geodesic distance function $d$ restricted
to $\sigma \times\sigma'$ for any two cells $\sigma, \sigma'\in \TD\cup\TDM$.
We call a cell $\sigma \in \TD \cup \TDM$ \emph{oceanic}
if $\sigma \subset \calM$, or \emph{coastal}, otherwise.
If both $\sigma, \sigma' \in \TD\cup\TDM$ are coastal,
then $\sigma, \sigma' \in \TD$ and the case is well understood
as discussed in Section~\ref{sec:td}. Otherwise, there are two cases:
the {\em ocean-to-ocean case} where both $\sigma$ and $\sigma'$ are
oceanic, and the {\em coast-to-ocean case} where only one of them is
oceanic. We discuss the two cases below.

%We thus focus on the case where at least one of $\sigma$ and $\sigma'$ is oceanic.

\subparagraph{Ocean-to-ocean}
For the ocean-to-ocean case, we extend the alignedness relation for all oceanic cells in $\TD \cup \TDM$.
To this end, when both $\sigma$ and $\sigma'$ are in $\TD$ or $\TDM$,
the alignedness has already been defined.
For any two oceanic cells $\sigma \in \TD$ and $\sigma'\in \TDM$, we
define their alignedness relation in the following way. If $\sigma$ is
contained in a cell $\sigma''\in\TDM$ that is aligned with $\sigma'$,
then we say that $\sigma$ and $\sigma'$ are \emph{aligned}. However,
$\sigma$ may not be contained in a cell of $\TDM$ because the
endpoints of horizontal diagonals of $\TDM$ that are on bay/canal
gates are not vertices of $\TD$ and those endpoints create vertical
diagonals in $\TDM$  that are not in $\TD$. To resolve this issue, we
augment $\TD$ by adding the vertical diagonals of $\TDM$ to $\TD$.
Specifically, for each vertical diagonal $l$ of $\TDM$, if no diagonal
in $\TD$ contains $l$, then we add $l$ to $\TD$ and extend $l$
vertically until it hits the boundary of $\calP$. In this way, we add
$O(h)$ vertical diagonals to $\TD$, and the size of $\TD$ is still
$O(n^2)$. Further, all results we obtained before are still applicable
to the new $\TD$. With a little abuse of notation, we still use $\TD$
to denote the new version of $\TD$. Now, for any two oceanic cells
$\sigma \in \TD$ and $\sigma'\in \TDM$, there must be a unique cell
$\sigma''\in\TDM$ that contains $\sigma$, and $\sigma$ and $\sigma'$
are defined to be {\em aligned} if and only if $\sigma''$ and
$\sigma'$ are aligned.
Lemmas~\ref{lem:key} and~\ref{lem:keym} are naturally extended as
follows, along with this extended alignedness relation.

\begin{lemma} \label{lem:oceanic}
Let $\sigma, \sigma'\in\TD\cup\TDM$ be two oceanic cells.
For any $s\in\sigma$ and $t\in\sigma'$,
it holds that $d(s,t) = |\seg{st}|$ if $\sigma$ and $\sigma'$ are aligned;
otherwise, there exists a shortest \st\ path that passes through
a vertex $v\in V_\sigma$ and a vertex $v' \in V_{\sigma'}$.
%\hfill\qed
\end{lemma}
\begin{proof}
If both of $\sigma$ and $\sigma'$ belong to $\TD$ or $\TDM$,
Lemmas~\ref{lem:key} and~\ref{lem:keym} are applied.
Suppose that $\sigma\in\TD$ and $\sigma'\in\TDM$.
If they are aligned, then $\sigma$ is contained in $\sigma''\in\TDM$ that is
aligned with $\sigma'$ by definition.
Hence, we have $d(s,t) = |\seg{st}|$ by Lemma~\ref{lem:keym}.
\end{proof}

\subparagraph{Coast-to-ocean}
We then turn to the coast-to-ocean case.
We now focus on a bay or canal $A$.
Since $A$ has gates, we need to somehow incorporate the influence of
its gates into the decomposition $\TD$. To this end,
we add $O(1)$ additional diagonals into $\TDM$ as follows:
extend a horizontal line from each endpoint of each gate of $A$ until it
hits $\bd \calM$, and then extend a vertical line from each endpoint of each gate of $A$
and each endpoint of the horizontal diagonals that are added above.
Let $\TDM^A$ denote the resulting decomposition of $\calM$.
Note that there are some cells of $\TDM$ each of which is partitioned into $O(1)$ cells of $\TDM^A$
but the combinatorial complexity of $\TDM^A$ is still $O(h^2)$.
For any gate $g$ of $A$,
let $C_g\subset\calP$ be the cross-shaped region of points in $\calP$
that can be joined with a point on $g$ by a vertical or horizontal line segment inside $\calP$.
Since the endpoints of $g$ are also obstacle vertices,
the boundary of $C_g$ is formed by four diagonals of $\TD$.
Hence, any cell in $\TD$ or $\TDM^A$ is either completely contained in $C_g$
or interior-disjoint from $C_g$.
A cell of $\TD$ or $\TDM^A$ in the former case is said to be \emph{$g$-aligned}.
%On the other hand, any cell $\sigma'\in\TDM$ is said to be \emph{$g$-aligned}
%if $\sigma'$ intersects $C_g$ in their interior.
%On the other hand, any cell $\sigma' \in \TDM^A$
%is said to be \emph{$g$-aligned} if $\sigma \subset  R_g$.

In the following, %Lemmas~\ref{lem:c2o-oo}--\ref{lem:c2o-xo},
we let $\sigma\in\TD$ be
any coastal cell that intersects $A$ and $\sigma'\in\TDM^A$ be any oceanic cell.
Depending on whether $\sigma$ and $\sigma'$ are $g$-aligned for a gate $g$ of $A$,
there are three cases: (1) both cells are $g$-aligned; (2) $\sigma'$ is not $g$-aligned; (3) $\sigma'$ is $g$-aligned but $\sigma$ is not. Lemma \ref{lem:c2o-oo} handles the first case. Lemma \ref{lem:c2o-oceanic} deals with a special case for the latter two cases.
Lemma \ref{lem:c2o-*x} is for the second case. Lemma \ref{lem:c2o-xo} is for the third case and Lemma \ref{lem:c2o_vg} is for proving Lemma \ref{lem:c2o-xo}. The proof of Lemma \ref{lem:dist_ex} summarizes the entire algorithm for all three cases.

\begin{lemma} \label{lem:c2o-oo}
Suppose that $\sigma$ and $\sigma'$ are both $g$-aligned for a gate $g$ of $A$.
Then, for any $s\in\sigma$ and $t\in\sigma'$,
we have $d(s,t) = |\seg{st}|$.
\end{lemma}
\begin{proof}
It suffices to observe that $s\in\sigma$ and $t\in\sigma'$ in $C_g$ can be joined by
an L-shaped rectilinear path, whose length is equal to the $L_1$ distance between them
by Fact~\ref{fact:l1length}.
\end{proof}

Consider any path $\pi$ in $\calP$ from $s\in \sigma$ to $t\in\sigma'$,
and assume $\pi$ is directed from $s$ to $t$.
For a gate $g$ of $A$, we call $\pi$ \emph{$g$-through}
if $g$ is the last gate of $A$ crossed by $\pi$. %, or if $\pi \cap A = \emptyset$.
The path $\pi$ is a {\em shortest $g$-through path} if its
$L_1$ length is the smallest among all $g$-through paths from $s$ to $t$.
Suppose $\pi$ is a shortest path from $s$ to $t$ in $\calP$. Since
$\sigma$ may intersect $\calM$, if $s\in \sigma$ is not in $A$, then
$\pi$ may {\em avoid} $A$ (i.e., $\pi$ does not intersect $A$). If $A$
is a bay, then either $\pi$ avoids $A$ or $\pi$ is a shortest
$g$-through path for the only gate $g$ of $A$; otherwise (i.e., $A$ is
a canal), either $\pi$ avoids $A$ or $\pi$ is a shortest $g$-through
or $g'$-through path for the two gates $g$ and $g'$ of $A$.
We have the following lemma, which is self-evident.

\begin{lemma} \label{lem:c2o-oceanic}
Suppose that for any gate $g$ of $A$, at least one of $\sigma$ and
$\sigma'$ is not $g$-aligned.
For any $s\in \sigma$ and $t \in \sigma'$,
if there exists a shortest \st\ path that avoids $A$,
then a shortest \st\ path passes through
a vertex $v\in V_\sigma$ and another vertex $v'\in V_{\sigma'}$.
\end{lemma}

We then focus on shortest $g$-through paths
according to the $g$-alignedness of $\sigma$ and $\sigma'$.

\begin{lemma} \label{lem:c2o-*x}
Suppose $\sigma'$ is not $g$-aligned for a gate $g$ of $A$ and there are no
shortest \st\ paths that avoid $A$. Then, for any $s\in\sigma$ and $t\in\sigma'$,
there exists a shortest $g$-through \st\ path
containing a vertex $v\in V_\sigma$ and a vertex
$v'\in V_{\sigma'}$.
\end{lemma}
\begin{proof}
%Note that since $\sigma$ is not $g$-aligned, $\sigma$ is not crossed by $g$.
If $A$ is a bay, since there are no shortest \st\ paths that avoid $A$, $s\in \sigma$
must be contained in $A$, and thus there must exist $g$-through paths
from $s\in\sigma$ and $t\in\sigma'$. If $A$ is a canal, although
$\sigma$ may be crossed by the other gate of $A$, there also exist
$g$-through paths from $s\in\sigma$ and $t\in\sigma'$. More
specifically, if $s\in A$, then there are $g$-through paths from $s$
to any $t\in \sigma'$; otherwise there are also $g$-through paths from
$s$ to any $t\in \sigma'$ that cross both gates of $A$.

Let $\pi$ be any shortest $g$-through path between $s\in\sigma$ and $t\in\sigma'$.
Since $\pi$ is $g$-through and $\sigma'$ is not $g$-aligned,
$\pi$ crosses
a horizontal and a vertical diagonals of $\TD$ that define $C_g$, and
escape $C_g$ to reach $t$ in $\sigma'$.
This implies that $\pi$ intersects a horizontal and a vertical diagonals defining $\sigma$
and thus can be modified to pass through a vertex $v \in V_\sigma$ of $\sigma$
as done in the proof of Lemma~\ref{lem:key}.
At the opposite end $t$, since $\sigma'\cap C_g = \emptyset$,
we can apply the above argument symmetrically to
modify $\pi$ to pass through a vertex $v' \in V_{\sigma'}$ of $\sigma'$.
Thus, the lemma follows.
\end{proof}

The remaining case is when $\sigma'\in\TDM^A$ is $g$-aligned %for a gate $g$ of $A$
but $\sigma\in\TD$ is not.
Recall $\sigma$ is coastal and intersects $A$, and $\sigma'$ is
oceanic (implying $\sigma'$ does not intersect $A$).

\begin{lemma} \label{lem:c2o_vg}
Let $g$ be a gate of $A$, and suppose that $\sigma$ is not $g$-aligned.
Then, there exists a unique vertex $v_g \in V_\sigma\cap A$ such that
for any $s\in\sigma$ and $x\in g$,
the concatenation of segment $\seg{sv_g}$ and any $L_1$ shortest path
from $v_g$ to $x$ inside $A\cup \sigma$
results in an $L_1$ shortest path from $s$ to $x$ in $A \cup \sigma$.
\end{lemma}
\begin{proof}
Let $P:= A\cup \sigma$. Since $\sigma$ is not $g$-aligned, $\sigma$ does not intersect
the gate $g$. Therefore, if $A$ is a bay, $\sigma$ must be contained in $A$ and thus $P=A$;
if $A$ is a canal, $\sigma$ may intersect the other gate of $A$
while the union $P$ forms a simple polygon.
Thus, $P$ is a simple polygon, and we apply Fact~\ref{fact:euc_simple} to $P$.
Let $\pi_2(s,t)$ be the unique Euclidean shortest path between $s, t\in P$ in $P$.
%the fact that
%in a simple polygon $A$, there is a unique Euclidean shortest path
%between any two points, and a Euclidean shortest path is also an $L_1$ shortest path.
Consider the union $H$ of $\pi_2(s,x)$ for all $s\in\sigma$ and all points $x\in g$,
and suppose $\pi_2(s,x)$ is directed from $s$ to $x$.  Then, $H$ forms an hourglass.
We distinguish two possibilities: either $H$ is open or closed.

\begin{figure}[t]
\begin{minipage}[t]{\linewidth}
\begin{center}
\includegraphics[totalheight=1.5in]{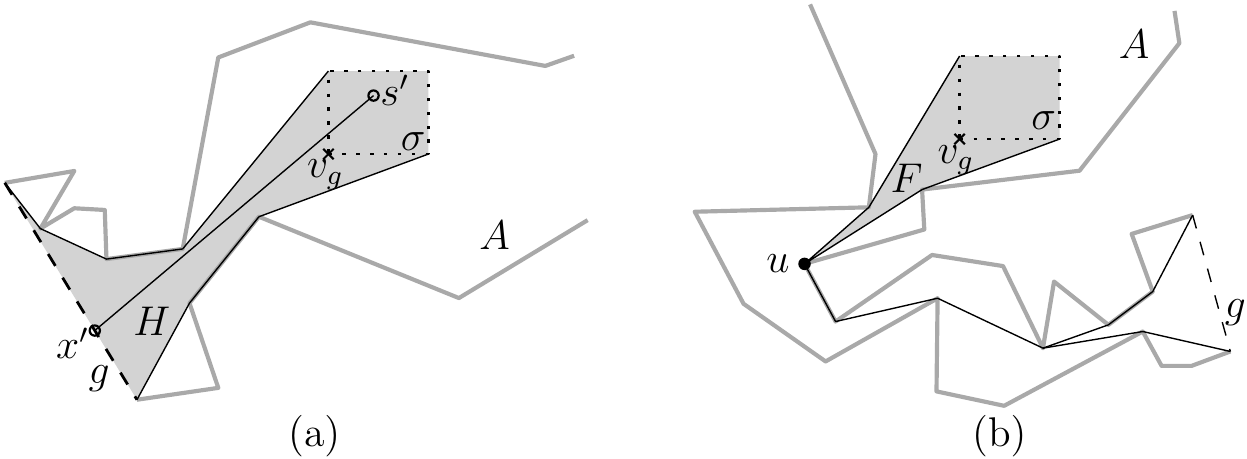}
\caption{\footnotesize Illustration to the proof of Lemma~\ref{lem:c2o_vg}. (a) When $H$ is open or
(b) closed.
}
\label{fig:vg}
\end{center}
\end{minipage}
\end{figure}

Assume that $H$ is open.  See \figurename~\ref{fig:vg}(a).
Then, there exist $s'\in\sigma$ and $x'\in g$ such that $\seg{s'x'} \subset P$,
and thus $\pi_2(s',x') = \seg{s'x'}$.
Without loss of generality, we assume that $s'$ lies to the right of and above $x'$
so that $\pi_2(s',x')$ is left-downwards.
We then observe that the first segment $\seg{su}$ of $\pi(s,x)$ for any $s\in\sigma$ and any $x\in g$
is also left-downwards since $\sigma$ is a cell of $\TD$ and is not $g$-aligned.
This implies that the shortest path $\pi_2(s, x)$ contained in $H$ crosses
the same pair of a vertical and a horizontal diagonals that define $\sigma$;
more precisely, it crosses the left vertical and the lower horizontal diagonals of $\TD$.
Letting $v_g$ be the vertex defined by the two diagonals,
we apply the same argument as in the proof of Lemma~\ref{lem:key}
to modify $\pi_2(s,x)$ to pass through $v_g$.

Now, assume that $H$ is closed.
See \figurename~\ref{fig:vg}(b).
Then, $H$ has two funnels and let $\funnel$ be the one that contains $\sigma$.
Let $u$ be the apex of $\funnel$, that is every Euclidean shortest path in $H$
passes through $u$.
Note that $u$ is an obstacle vertex, and thus
$u \in V_{\sigma''}$ for a cell $\sigma'' \in \TD$ that is not aligned with $\sigma$.
Without loss of generality, we assume that $\sigma$ lies to the right of and above $\sigma'$.
We observe that the Euclidean shortest path $\pi_2(s, u)$ for any $s\in\sigma$ is monotone since $u$ is the apex of $\funnel$.
Since $u$ lies to the left of and below any $s\in\sigma$ and
$\pi_2(s, u)$ is monotone,
we can modify the path to pass through the bottom-left vertex $v_g\in V_\sigma$,
as in the previous case.
\end{proof}

From now on, let $v_g$ be the vertex as described in Lemma~\ref{lem:c2o_vg}
($v_g$ can be found efficiently, as shown in the proof of Lemma~\ref{lem:dist_ex}).
Consider the union of the Euclidean shortest paths inside $A$ from $v_g$
to all points $x \in g$.  Since $A$ is a simple polygon,
the union forms a funnel $\funnel_g(v_g)$ with base $g$,
plus the Euclidean shortest path from $v_g$ to the apex of $\funnel_g(v_g)$.
Recall Fact~\ref{fact:euc_simple} that any Euclidean shortest path inside a simple polygon
is also an $L_1$ shortest path.
Let $W_g(v_g)$ be the set of horizontally and vertically
extreme points in each convex chain of $\funnel_g(v_g)$,
that is, $W_g(v_g)$ gathers the leftmost, rightmost, uppermost, and lowermost points
in each chain of $\funnel_g(v_g)$.
Note that $|W_g(v_g)| \leq 8$ %\footnote{%
%In fact, one can see that $|W_g(v_g)| \leq 5$ by a careful analysis.
%For our purpose, it is sufficient to see that $|W_g(v_g)| = O(1)$.}
and $W_g(v_g)$ includes the endpoints of $g$ and the apex of $\funnel_g(v_g)$.
We then observe the following lemma.

\begin{lemma} \label{lem:c2o-xo}
Suppose that $\sigma'$ is $g$-aligned but $\sigma$ is not.
Then, for any $s\in\sigma$ and $t\in\sigma'$,
there exists a shortest $g$-through \st\ path that passes through
$v_g$ and some $w\in W_g(v_g)$.
Moreover, the length of such a path is $|\seg{sv_g}| + d(v_g, w) + |\seg{wt}|$.
\end{lemma}
\begin{proof}
Since $A$ is a simple polygon,
any Euclidean shortest path in $A$ is also an $L_1$ shortest path
by Fact~\ref{fact:euc_simple}.
Thus, the $L_1$ length of a shortest path
from $v_g$ to any point $x$ in the funnel $\funnel_g(v_g)$ is equal to
the $L_1$ length of the unique Euclidean shortest path in $A$,
which is contained in $\funnel_g(v_g)$.

By Lemma~\ref{lem:c2o_vg} and the assumption that $\sigma'$
is $g$-aligned, among the paths from $s$ to $t$ that cross the gate
$g$,
there exists an $L_1$ shortest $g$-through \st\ path $\pi$ consisting of three portions:
$\seg{sv_g}$, the unique Euclidean shortest path from $v_g$ to a vertex $u$
on a convex chain of $\funnel_g(v_g)$, and $\seg{ut}$.
Let $w\in W_g(v_g)$ be the last one among $W_g(v_g)$ that we encounter
during the walk from $s$ to $t$ along $\pi$.
Consider the segment $\seg{wt}$, which may cross $\bd \funnel_g(v_g)$.
If $\seg{wt} \cap \bd \funnel_g(v_g) = \emptyset$, then we are done
by replacing the subpath of $\pi$ from $u$ to $t$ by $\seg{wt}$.
Otherwise, $\seg{wt}$ crosses $\bd \funnel_g(v_g)$ at two points $p, q\in\bd \funnel_g(v_g)$.
Since $W_g(v_g)$ includes all extreme points of each chain of $\funnel_g(v_g)$,
there is no $w'\in W_g(v_g)$ on the subchain of $\funnel_g(v_g)$ between $p$ and $q$.
Hence, we can replace the subpath of $\pi$ from $w$ to $t$ by
a monotone path from $w$ to $t$, which consists of
$\seg{wp}$, the convex path from $p$ to $q$ along $\bd \funnel_g(v_g)$,
and $\seg{qt}$, and the $L_1$ length of the above monotone path
is equal to $|\seg{wt}|$ by Fact~\ref{fact:l1length}.
Consequently, the resulting path is also an $L_1$ shortest path with the desired property.
\end{proof}

For any cell $\sigma \in \TD\cup \TDM$, let $n_\sigma$ be the combinatorial complexity of $\sigma$.
If $\sigma$ is a boundary cell of $\TDM$, then
$n_\sigma$ may not be bounded by a constant;
otherwise, $\sigma$ is a trapezoid or a triangle, and thus $n_\sigma \leq 4$.
The geodesic distance function $d$ defined on $\sigma\times\sigma'$
for any two cells $\sigma, \sigma'\in \TD\cup\TDM$
can be explicitly computed in $O(n_\sigma n_{\sigma'})$ time
after some preprocessing, as shown in Lemma \ref{lem:dist_ex}.

%\begin{lemma} \label{lem:dist_ex1}
% Let $\sigma, \sigma'\in \TD \cup \TDM$, and suppose that
% both are oceanic or both are coastal.
% The geodesic distance function $d$ on $\sigma\times\sigma'$
% can be explicitly computed in $O(n_\sigma n_{\sigma'})$ time,
% provided that $d(v, v')$ is known for any $v\in V_\sigma$ and $v'\in V_{\sigma'}$.
%% where $n_\sigma$ and $n_\sigma'$ denote the complexity of cells $\sigma$
%% and $\sigma'$, respectively.
%\end{lemma}
%\begin{proof}
%If both of $\sigma$ and $\sigma'$ are oceanic, we have $\sigma, \sigma'\in\TDM$,
%while we have $\sigma, \sigma'\in\TD$ in the other case.
%In the oceanic case,
%Lemma~\ref{lem:oceanic} implies that for any $(s,t)\in\sigma\times\sigma'$,
%$d(s,t) = |\seg{st}|$ if they are aligned, or
%$d(s,t) = \min_{v\in V_\sigma, v'\in V_{\sigma'}} d_{vv'}(s,t)$,
%where $d_{vv'}(s,t) = |\seg{sv}| + d(v, v') + |\seg{v't}|$.
%On the other hand, in the coastal case, Lemma~\ref{lem:key} implies the same conclusion.
%
%Since $|V_\sigma| \leq 4$ and $|V_{\sigma'}| \leq 4$ in either case,
%the geodesic distance $d$ on $\sigma\times\sigma$
%is the lower envelope of at most $16$ linear function.
%Hence, provided that the values of $d(v,v')$ for all pairs $(v,v')$ are known,
%the envelope can be computed in time proportional to the complexity of the domain
%$\sigma \times \sigma'$, which is $O(n_\sigma n_{\sigma'})$.
%\end{proof}

\begin{lemma} \label{lem:dist_ex}
Let $\sigma$ be any cell of $\TD$ or $\TDM$.  After $O(n)$-time preprocessing,
the function $d$ on $\sigma\times\sigma'$ for
any cell $\sigma'\in\TD\cup\TDM$ can be explicitly computed in $O(n_\sigma n_{\sigma'})$ time,
provided that $d(v, v')$ has been computed for any $v\in V_\sigma$ and any $v'\in V_{\sigma'}$.
Moreover, $d$ on $\sigma \times \sigma'$ is the lower envelope of $O(1)$ linear functions.
\end{lemma}
\begin{proof}
If both $\sigma$ and $\sigma'$ are oceanic, then
Lemma~\ref{lem:oceanic} implies that for any $(s,t)\in\sigma\times\sigma'$,
$d(s,t) = |\seg{st}|$ if they are aligned, or
$d(s,t) = \min_{v\in V_\sigma, v'\in V_{\sigma'}} d_{vv'}(s,t)$,
where $d_{vv'}(s,t) = |\seg{sv}| + d(v, v') + |\seg{v't}|$.
On the other hand, if $\sigma$ and $\sigma'$ are coastal,
then both are cells of $\TD$ and
Lemma~\ref{lem:key} implies the same conclusion.
Since $|V_\sigma| \leq 4$ and $|V_{\sigma'}| \leq 4$ in either case,
the geodesic distance $d$ on $\sigma\times\sigma$
is the lower envelope of at most $16$ linear functions.
Hence, provided that the values of $d(v,v')$ for all pairs $(v,v')$ are known,
the envelope can be computed in time proportional to the complexity of the domain
$\sigma \times \sigma'$, which is $O(n_\sigma n_{\sigma'})$.

From now on, suppose that $\sigma$ is coastal and $\sigma'$ is oceanic.
Then, $\sigma$ is a cell of $\TD$ and intersects some bay or canal $A$.
If $\sigma'$ is also a cell of $\TD$, then Lemma~\ref{lem:key} implies the lemma,
as discussed in Section~\ref{sec:td}; thus, we assume $\sigma'$ is a cell of $\TDM$.

As above, we add diagonals extended from each endpoint of each gate of $A$ to obtain $\TDM^A$,
and specify all $g$-aligned cells for each gate $g$ of $A$ in $O(n)$ time.
In the following, let $\sigma'$ be an oceanic cell of $\TD$ or of $\TDM^A$.
Note that a cell of $\TDM$ can be partitioned into $O(1)$ cells of $\TDM^A$.
We have two cases depending on whether $A$ is a bay or a canal.

First, suppose that $A$ is a bay; let $g$ be the unique gate of $A$.
In this case, any $L_1$ shortest path is $g$-through,
provided that it intersects $A$, since $g$ is unique.
There are two subcases depending on whether $\sigma$ is $g$-aligned or not.

\begin{itemize}
\item
If $\sigma$ is $g$-aligned, then by Lemmas~\ref{lem:c2o-oo}, \ref{lem:c2o-oceanic}, and~\ref{lem:c2o-*x},
we have $d(s,t) = |\seg{st}|$ if $\sigma'$ is $g$-aligned,
or $d(s,t) = \min_{v\in V_\sigma, v'\in V_{\sigma'}} d_{vv'}(s,t)$, otherwise,
where $d_{vv'}(s,t) = |\seg{sv}| + d(v,v') + |\seg{v't}|$.
Thus, the lemma follows by an identical argument as above.
\item
Suppose that $\sigma$ is not $g$-aligned.
Then, $\sigma \subset A$ since $A$ has a unique gate $g$.
In this case, we need to find the vertex $v_g\in V_\sigma$.
For the purpose, we compute at most four Euclidean shortest path maps $\spm_A(v)$
inside $A$ for all $v\in V_\sigma$ in $O(n)$ time~\cite{ref:GuibasLi87}.
By Fact~\ref{fact:euc_simple}, $\spm_A(v)$ is also an $L_1$ shortest path map in $A$.
We then specify the $L_1$ geodesic distance from $v$ to all points on $g$,
which results in a piecewise linear function $f_v$ on $g$.
For each $v\in V_\sigma$, we test whether it holds that
$f_v(x) + |\seg{vv'}| \leq f_{v'}(x)$ for all $x\in g$ and all $v'\in V_\sigma$.
By Lemma~\ref{lem:c2o_vg}, there exists a vertex in $V_\sigma$
for which the above test is passed, and such a vertex is $v_g$.
Since each shortest path map $\spm_A(v)$ is of $O(n)$ complexity,
all the above effort to find $v_g$ is bounded by $O(n)$.
Next, we compute the funnel $\funnel_g(v_g)$ and
the extreme vertices $W_g(v_g)$ as done above
by exploring $\spm_A(v_g)$ in $O(n)$ time.

If $\sigma'$ is not $g$-aligned, we apply Lemma~\ref{lem:c2o-*x}
to obtain $d(s,t) = \min_{v\in V_\sigma, v'\in V_{\sigma'}} d_{vv'}(s,t)$.
Thus, $d$ is the lower envelope of at most $16$ linear functions over $\sigma\times\sigma'$.
Otherwise, if $\sigma'$ is $g$-aligned, then we have
$d(s,t) = \min_{w\in W_g(v_g)} d_{v_gw}(s,t)$ by Lemma~\ref{lem:c2o-xo}.
Since $|W_g(v_g)| \leq 8$,
$d$ is the lower envelope of a constant number of linear functions.
\end{itemize}
Thus, in any case, we conclude the bay case.

Now, suppose that $A$ is a canal.
Then, $A$ has two gates $g$ and $g'$, and $\sigma$ falls into one of the three case:
(i) $\sigma$ is both $g$-aligned and $g'$-aligned,
(ii) $\sigma$ is neither $g$-aligned nor $g'$-aligned, or
(iii) $\sigma$ is $g$- or $g'$-aligned but not both.
As a preprocessing, if $\sigma$ is not $g$-aligned, then
we compute $v_g$, $\funnel_g(v_g)$, and $W_g(v_g)$ as done in the bay case;
analogously, if not $g'$-aligned, compute $v_{g'}$, $\funnel_{g'}(v_{g'})$, and $W_{g'}(v_{g'})$.
Note that any shortest path in $\calP$ is either $g$-through or $g'$-through,
provided that it intersects $A$.
Thus, $d(s,t)$ chooses the minimum among a shortest $g$-through path,
a shortest $g'$-through path, and a shortest path avoiding $A$ if possible.
We consider each of the three cases of $\sigma$.
\begin{enumerate}%[(i)] %\denseitems
\item Suppose that $\sigma$ is both $g$-aligned and $g'$-aligned.
 In this case, if $\sigma'$ is either $g$-aligned or $g'$-aligned,
 then we have $d(s,t) = |\seg{st}|$ by Lemma~\ref{lem:c2o-oo}.
 Otherwise, if $\sigma'$ is neither $g$-aligned nor $g'$-aligned,
 then we apply Lemmas~\ref{lem:c2o-oceanic} and~\ref{lem:c2o-*x} to have
 $d(s,t) = \min_{v\in V_\sigma, v'\in V_{\sigma'}} d_{vv'}(s,t)$.
 Hence, the lemma follows.
\item Suppose that $\sigma$ is neither $g$-aligned nor $g'$-aligned.
 If $\sigma'$ is both $g$-aligned and $g'$-aligned, then by Lemma~\ref{lem:c2o-xo}
  the length of a shortest $g$-through path is equal to
 $\min_{w\in W_g(v_g)} d_{v_gw}(s,t)$ while
 the length of a shortest $g'$-through path is equal to
 $\min_{w\in W_{g'}(v_{g'})} d_{v_{g'}w}(s,t)$.
 The geodesic distance $d(s,t)$ is the minimum of the above two quantities,
 and thus the lower envelope of $O(1)$ linear function on $\sigma\times\sigma'$.

 If $\sigma'$ is $g$-aligned but not $g'$-aligned, then
 by Lemmas~\ref{lem:c2o-*x} and~\ref{lem:c2o-xo},
 we have
 \[ d(s,t) = \min \{ \min_{w\in W_g(v_g)} d_{v_gw}(s,t), \min_{v\in V_\sigma, v'\in V_{\sigma'}} d_{vv'}(s,t) \}.\]
 The case where $\sigma'$ is $g'$-aligned but not $g$-aligned is analogous.

 If $\sigma'$ is neither $g$-aligned nor $g'$-aligned,
 then
 $d(s,t) = \min_{v\in V_\sigma, v'\in V_{\sigma'}} d_{vv'}(s,t)$ by Lemma~\ref{lem:c2o-*x}.

\item Suppose that $\sigma$ is $g'$-aligned but not $g$-aligned.
 The other case where it is $g$-aligned but not $g'$-aligned can be handled symmetrically.
 If $\sigma'$ is $g'$-aligned, then we have $d(s,t) =|\seg{st}|$ by Lemma~\ref{lem:c2o-oo}.
 If $\sigma'$ is neither $g$-aligned nor $g'$-aligned,
 then, by Lemmas~\ref{lem:c2o-oceanic} and~\ref{lem:c2o-*x},
 $d(s,t) = \min_{v\in V_\sigma, v'\in V_{\sigma'}} d_{vv'}(s,t)$.

 The remaining case is when $\sigma'$ is $g$-aligned but not $g'$-aligned.
 In this case, the length of a shortest $g$-through path is equal to
 $\min_{w\in W_g(v_g)} d_{v_gw}(s,t)$ by Lemma~\ref{lem:c2o-xo} for gate $g$
 while the length of a shortest $g'$-through path is equal to
 $\min_{v\in V_\sigma, v'\in V_{\sigma'}} d_{vv'}(s,t)$ by Lemmas~\ref{lem:c2o-oceanic} and~\ref{lem:c2o-*x}.
 Thus, the geodesic distance $d(s,t)$ is the smaller of the two quantities.
\end{enumerate}
Consequently, we have verified every case of $(\sigma, \sigma')$.

As the last step of the proof,
observe that it is sufficient to handle separately all the cells $\sigma' \in \TDM^A$
whose union forms the original cell of $\TDM$,
since any cell of $\TDM$ can be decomposed into $O(1)$ cells of $\TDM^A$.
\end{proof}

\subsection{Computing the Geodesic Diameter and Center}

Lemma~\ref{lem:bay} assures that
we can ignore coastal cells that are completely contained in the interior of a bay or canal,
in order to find a farthest point from any $s\in \calP$.
This suggests a combined set $\TDF$ of cells from the two different decompositions
$\TD$ and $\TDM$:
%We will call such a cell $\sigma \in \TD \cup \TDM$ with $\bd\sigma\cap\bd\calP \neq\emptyset$
%a \emph{boundary cell}.
Let $\TDF$ be the set of all cells $\sigma$ such that
either $\sigma$ belongs to $\TDM$ or $\sigma \in \TD$ is a coastal cell
with $\bd \sigma \cap \bd \calP \neq \emptyset$.
Note that $\TDF$ consists of $O(h^2)$ oceanic cells from $\TDM$ and $O(n)$ coastal cells
from $\TD$.
Since the boundary $\bd A$ of any bay or canal $A$ is covered by the cells of $\TDF$,
Lemma~\ref{lem:bay} implies the following lemma.

\begin{lemma} \label{lem:combined_cell}
For any point $s\in\calP$, %it holds that
$\max_{t\in \calP} d(s,t) = \max_{\sigma'\in \TDF} \max_{t\in \sigma'} d(s,t)$. %\hfill\qed
\end{lemma}

We apply the same approach as in Section~\ref{sec:td} but we use $\TDF$ instead of $\TD$.
%Lemma~\ref{lem:combined_cell} implies that the diameter is equal to the maximum of the $(\sigma, \sigma')$-constrained diameter for all pairs of $\sigma, \sigma'\in \TDF$; a $\sigma$-constrained center for any cell $\sigma \in \TD$ is attained by a point $c\in \sigma$ minimizing $\max_{\sigma'\in\TDF} \max_{t\in \sigma'} d(c,t)$.

%\subsection{Computing the diameter}
To compute the $L_1$ geodesic diameter,
we compute the $(\sigma, \sigma')$-constrained diameter
for each pair of cells $\sigma,\sigma'\in\TDF$.
%The correctness is guaranteed by Lemma~\ref{lem:combined_cell}.
Suppose we know the value of $d(v, v')$ for any $v\in V_\sigma$ and
any $v\in V_{\sigma'}$ over all $\sigma, \sigma' \in \TDF$.
Our algorithm handles each pair $(\sigma, \sigma')$ of cells in $\TDF$
according to their types by applying Lemma~\ref{lem:dist_ex}.
The following lemma computes $d(v, v')$ for all cell vertices $v$ and $v'$ of $\TDF$.

\begin{lemma} \label{lem:dist_vv}
 In $O(n^2 + h^4)$ time, one can compute
 the geodesic distances $d(v, v')$ between every $v\in V_\sigma$ and $v' \in V_{\sigma'}$
 for all pairs of two cells $\sigma, \sigma'\in\TDF$.
\end{lemma}
\begin{proof}
Let $V_\calM$ be the set of vertices $V_\sigma$ for all oceanic cells $\sigma \in \TDM$,
and $V_c$ be the set of vertices $V_\sigma$ for all coastal cells $\sigma \in \TDF$.
Note that $|V_\calM| = O(h^2)$ and $|V_c| = O(n)$.
We handle the pairs $(v,v')$ of vertices separately in three cases:
(i) when $v, v'\in V_\calM$, (ii) when $v\in V_c$ and $v'\in V_\calM$,
and (iii) when $v, v'\in V_c$.

Let $v\in V_\calM$ be any vertex.
We compute the shortest path map $\spm_{core}(v)$ in the core domain $\calP_{core}$
as discussed in Section~\ref{sec:corridor}.
Recall that $\spm_{core}(v)$ is of $O(h)$ complexity
and can be computed in $O(n + h\log h)$ time (after $\calP$ is
triangulated)
\cite{ref:ChenA11ESA,ref:ChenL113STACS}.
For any point $p\in \calM$, the geodesic distance $d_{core}(v,p)$ can be determined in
constant time after locating the region of $\spm_{core}(v)$ that contains $p$.
By Lemma~\ref{lem:core}, we have $d(v, p) = d_{core}(v,p)$.
Computing $d(v, v')$ for all $v' \in V_\calM$ can be done in $O(h^2)$ time
by running the plane sweep algorithm of Lemma~\ref{lem:findcells}
on $\spm_{core}(v)$.
Thus, for each $v \in V_\calM$ we spend $O(n + h^2)$ time.
Since $|V_\calM| = O(h^2)$,
we can compute $d(v, v')$ for
all pairs of vertices $v, v' \in V_\calM$ in time $O(nh^2+h^4)$.

Case (ii) can also be handled in a similar fashion.
Let $v\in V_c$.
We also show that computing $d(v, v')$ for all $v'\in V_\calM$ can be
done in $O(n+h^2)$ time.
If $v$ lies in the ocean $\calM$, then we can apply the same argument as in case (i).
Thus, we assume $v \notin\calM$.
For the purpose, we consider $v$ as a point hole (i.e., a hole or an obstacle consisting of only one point) into the polygonal domain $\calP$
to obtain a new domain $\calP_v$,
and compute the corresponding corridor structures of $\calP_v$, which can be done in $O(n + h\log^{1+\epsilon} h)$ time (or $O(n)$ time after a triangulation  of $\calP$ is given), as discussed in Section~\ref{sec:corridornew}.
Let $\calM_v$ denote the ocean corresponding to the new polygonal domain $\calP_v$.
Since $v$ lies in a bay or canal of $\calP$,
$\calM$ is a subset of $\calM_v$ by the definition of bays, canals and the ocean.
Thus, we have $V_\calM \subset \calM_v$.
We then compute the core structure of $\calP_v$ and
the shortest path map $\spm_{core}'(v)$ in $\calM_v$ in $O(n + h\log h)$
time~\cite{ref:ChenA11ESA,ref:ChenL113STACS}.
Analogously, the complexity of $\spm_{core}'(v)$ is bounded by $O(h)$.
Finally, perform the plane sweep algorithm on $\spm_{core}'(v)$ as in case (i)
to get  the values of $d(v, v')$ for all $v' \in V_\calM$ in $O(h^2)$ time.

What remains is case (iii) where $v, v'\in V_c$.
Fix $v\in V_c$.  The vertices in $V_c$ either lie on $\bd \calP$ or in its interior.
In this case, we assume that we have a triangulation of $\calP$
as discussed in Section~\ref{sec:corridor}.  Recall that we
can compute the shortest path map $\spm(v)$ in $O(n + h\log h)$
time \cite{ref:ChenA11ESA,ref:ChenL113STACS}.
Since $\spm(v)$ stores $d(v, v')$ for all obstacle vertices $v'$ of $\calP$,
computing $d(v,v')$ for all $v' \in V_c \cap \bd \calP$ can be done in the same time bound
by adding $V_c \cap \bd \calP$ into $\calP$ as obstacle vertices.
Thus, the case where $v, v'\in V_c$ and one of them lies in $\bd \calP$
can be handled in $O(n^2 + nh\log h)$ time.

\begin{figure}[t]
\begin{minipage}[t]{\linewidth}
\begin{center}
\includegraphics[width=0.95\textwidth]{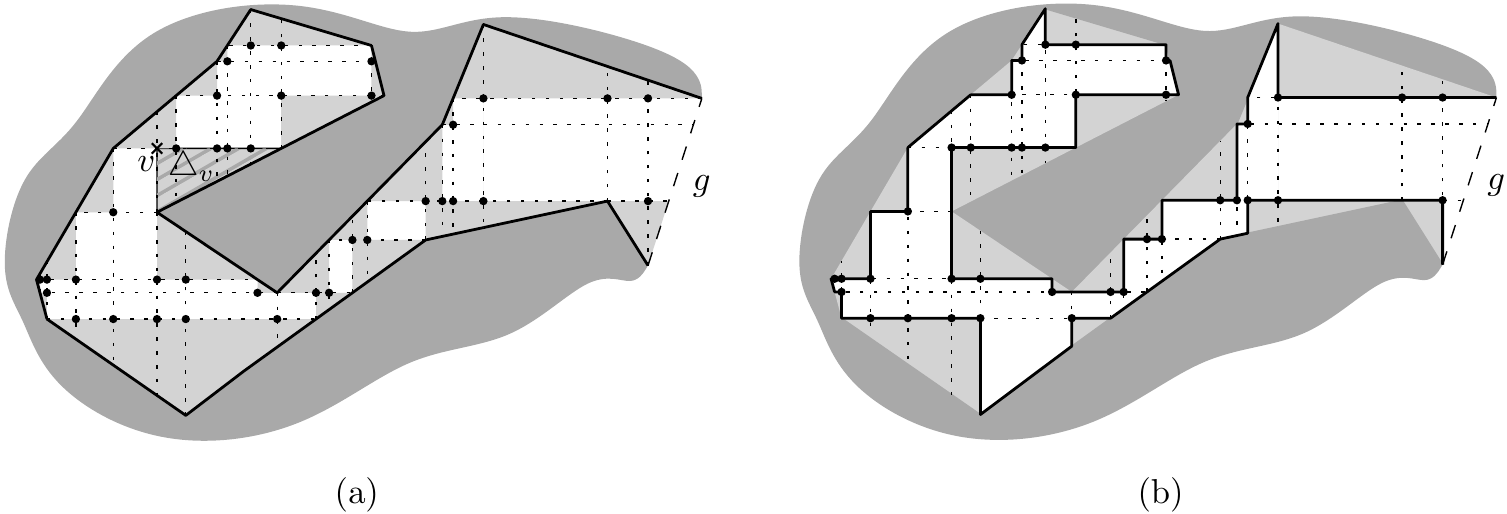}
\caption{\footnotesize Illustration to how to construct $\calP_{rect}$.
(a) A bay $A$ with gate $g$ and the decomposition $\TD$ inside $A$.
 The dark gray region depicts the hole of $\calP$ bounding $A$,
 coastal cells of $\TD$ intersecting $\bd A$ are shaded by light gray color,
 and black dots are vertices in $V_c\setminus \bd \calP$.
 One of those in $V_c$ is labeled $v$ and its triangle $\triangle_v$ is highlighted.
(b) The rectified polygonal domain $\calP_{rect}$ obtained by expanding the hole into
 $\triangle_v$ for all $v\in V_c$. The boundary of $\calP_{rect}$ is
 depicted by solid segments.
 Here, the edges of $A$ are ordered in the counter-clockwise order.
}
\label{fig:poly_rect}
\end{center}
\end{minipage}
\end{figure}

In the following, we focus on how to compute $d(v,v')$
for all $v,v' \in V_c\setminus \bd\calP$.
%Let $P_1, \ldots, P_h$ be the holes of $\calP$ and $P_0$ be the outer polygon of $\calP$.
%For simplicity of discussion, we shall call each of $P_0, \ldots, P_h$ a hole of $\calP$
%in this proof.
Let $e_1, \ldots, e_n$ be the edges of $\calP$ in an arbitrary order.
We modify the original polygonal domain $\calP$ to obtain
the \emph{rectified polygonal domain} $\calP_{rect}$ as follows.
For each $i=1, \ldots, n$, we define $V_i$ to be the set of all vertices
$v\in (V_c\setminus \bd \calP) \setminus (V_1 \cup \cdots\cup V_{i-1})$
such that $v\in V_\sigma$ for some coastal cell $\sigma\in\TDF$
with $\bd \sigma \cap e_i \neq \emptyset$.
%For $i = 0,\ldots, h$,
%define $V_i$ to be the set of
%all vertices $v\in (V_c\setminus \bd \calP) \setminus (V_0 \cup \cdots\cup V_{i-1})$
%such that $v \in V_\sigma$ for some coastal cell $\sigma\in \TDF$
%and $\bd \sigma$ intersects $\bd P_i$.
For each $v\in V_i$, we shoot two rays from $v$, vertical and horizontal, towards $e_i$
until each hits $e_i$.
Let $\triangle_v$ be the triangle formed by $v$ and the two points on $e_i$
hit by the rays.
Since $v$ is a vertex of a cell of $\TD$ facing $e_i$,
by the construction of $\TD$,
the two rays must hit $e_i$ and thus the triangle $\triangle_v$ is well defined.
%For each $v\in V_i$, we define the triangle $\triangle_v$ of $v$ as follows:
%if there is a unique coastal cell $\sigma\in \TDF$ such that
%$v \in V_\sigma$ and $\bd \sigma$ intersects $\bd P_i$,
%then we shoot two rays from $v$, vertical and horizontal, towards $P_i$,
%until each hits $\bd P_i$,
%then $\triangle_v$ is the triangle formed by $v$ and the points on $\bd P_i$
%hit by the rays.
%If there are two or more such cells $\sigma$,
%then $v$ must lie on a diagonal $e$ of $\TD$ extended from a reflex vertex $r$ of $P_i$.
%In this case, there is another diagonal $e'$ of $\TD$ such that $e \cap e' = v$,
%and let $\triangle_v$ be the triangle formed by the $r$ and the diagonal $e'$.
%Then, we take the union $Q_i$ of the closure of all coastal cells $\sigma \in \TDF$
%such that $\bd \sigma$ intersects $\bd P_i$ and
%there is $v\in V_i \cap V_\sigma$.
%We expand each hole $P_i$ of $\calP$ into $Q_i$,
%and denote by $\calP_{rect}$ the resulting polygonal domain with the expanded holes.
We then expand each hole of $\calP$ into the triangles $\triangle_v$ for all $v\in V_c\setminus \bd\calP$.
Let $\calP_{rect}$ be the resulting polygonal domain;
that is, every triangle $\triangle_v$ is regarded as an obstacle in $\calP_{rect}$.
%Then, we expand each hole $P_i$ of $\calP$ into $\bigcup_{v\in V_i} \triangle_v$,
%and denote by $P'_i$ the resulting hole, that is, $P'_i = P_i \cup \bigcup_{v\in V_i} \triangle_v$.
%Let $\calP_{rect}$ be the polygonal domain with the holes $P'_0, \ldots, P'_h$.
We also add all those in $V_c \setminus \bd\calP$ into $\calP_{rect}$
as obstacle vertices.
See \figurename~\ref{fig:poly_rect}.
Observe that $\calP_{rect}$ is a subset of $\calP$ as subsets of $\Real^2$
and all those in $V_c\setminus \bd\calP$ lie on the boundary of $\calP_{rect}$
as its obstacle vertices.
%Also, $\calP_{rect}$ has $h$ holes as $\calP$ does.
For any two points $s,t\in \calP_{rect}$, let $d_{rect}(s, t)$ be the $L_1$ geodesic distance
between $s$ and $t$ in $\calP_{rect}$.

We then claim the following:
\begin{quote}
 \textit{For any $s, t\in \calP_{rect}$, it holds that $d(s, t) = d_{rect}(s,t)$.}
\end{quote}
Suppose that the claim is true.
The construction of $\calP_{rect}$ can be done in $O(n^2)$ time.
Then, for any $v\in V_c\setminus \bd \calP$,
we compute the shortest path map $\spm_{rect}(v)$ in the rectified domain $\calP_{rect}$
and obtain $d(v, v')$ for all other $v'\in V_c\setminus \bd \calP$.
Since $\calP_{rect}$ has $h$ holes and $O(n)$ vertices by our construction of $\calP_{rect}$,
this task can be done in $O(n + h\log h)$ time
\cite{ref:ChenA11ESA,ref:ChenL113STACS}.
At last, case (iii) can be processed in total $O(n^2 + nh\log h)$ time.

We now prove the claim, as follows.

%\begin{proof}[\textbf{Proof of the claim}]
\vspace{0.1in}
\noindent\textbf{Proof of the claim.}
For any $v\in V_c\setminus \bd \calP$, the triangle $\triangle_v$ is well defined.
We call a triangle $\triangle_v$ \emph{maximal}
if there is no other $\triangle_{v'}$ with $\triangle_v \subset \triangle_{v'}$.
Note that any two maximal triangles are interior-disjoint but
may share a portion of their sides.
Indeed, $\calP \setminus \calP_{rect}$ is the union of all maximal triangles $\triangle_v$.
Pick any connected component $C$ of $\calP \setminus \calP_{rect}$.
The set $C$ is either a maximal triangle $\triangle_v$ itself
or the union of two maximal triangles that share a portion of their sides
by our construction of $\calP_{rect}$ and of $\TD$.
In either case, observe that the portion $\bd C \setminus \bd \calP$ is a monotone path.
%If $R$ is a maximal triangle, then it is obvious that $\bd R \setminus \bd P_i$ is monotone.
%Otherwise, suppose that $R$ is the union of two maximal triangles $\triangle_v$ and $\triangle_v'$.
%Then,

Consider any $s, t\in \calP_{rect}$ and any shortest \st\ path $\pi$ in $\calP$,
that is, $d(s,t) = |\pi|$.
If $\pi$ lies inside $\calP_{rect}$, then we are done
since the $L_1$ length of any \st\ path inside $\calP_{rect}$ is at least $d(s,t)$.
Otherwise, $\pi$ may cross a number of connected components of $\calP\setminus \calP_{rect}$.
Pick any such connected component $C$ that is crossed by $\pi$.
Let $p$ and $q$ be the first and the last points on $\bd C$
we encounter when walking along $\pi$ from $s$ to $t$.
Since $\bd C\setminus \bd \calP$ is monotone as observed above,
the path $\pi_{pq}$ along the boundary of $C$ is also monotone.
By Fact~\ref{fact:l1length}, $|\pi_{pq}| = |\seg{pq}|$
and thus we can replace the subpath of $\pi$ between $p$ and $q$ by $\pi_{pq}$
without increasing the $L_1$ length.
The resulting path thus has length equal to $|\pi|$ and avoids the interior of $C$.
We repeat the above procedure for all such connected components $C$ crossed by $\pi$.
At last, the final path $\pi'$ has length equal to $|\pi|$ and avoids
the interior of $\calP\setminus \calP_{rect}$.
That is, $\pi'$ is a \st\ path in $\calP_{rect}$ with $|\pi'| = |\pi| = d(s,t)$.
Since $d_{rect}(s,t) \geq d(s,t)$ in general,
$\pi'$ is an $L_1$ shortest \st\ path in $\calP_{rect}$,
and hence $d_{rect}(s,t) = d(s,t)$.

This proves the above claim.
\vspace{0.1in}

Consequently, the total time complexity is bounded by
\[  O(nh^2 + h^4) + O(n^2 + nh\log h) = O(n^2 + nh^2 + h^4) = O(n^2 + h^4).\]

The lemma thus follows.
\end{proof}

%Our algorithm for the diameter
%first runs the procedure of Lemma~\ref{lem:dist_vv} in $O(n^2+h^4)$ time
%as a preprocessing step.
%Then, we handle each $\sigma \in \TDF$ to compute the $(\sigma, \sigma')$-constrained diameter for all $\sigma'\in\TDF$. By Lemma~\ref{lem:dist_ex}, for $\sigma\in\TDF$ we spend $O(n)$ time
%for another preprocessing and then $O(n+h^2)$ time to iterate all $\sigma'\in\TDF$.
%The correctness of the algorithm follows from Lemma~\ref{lem:combined_cell}.
%
%For computing the $L_1$ geodesic center, we consider $O(n^2)$ cells $\sigma\in\TD$ and compute
%all the $\sigma$-constrained centers.
%%As a preprocessing, we spend $O(n^4)$ time to compute
%%the geodesic distances $d(v,v')$ for all pairs of vertices of $\TD$ by Lemma~\ref{lem:findcells}.
%Fix a cell $\sigma \in \TD$.
%For all $\sigma' \in \TDF$, we compute the geodesic distance function
%$d$ restricted to $\sigma\times \sigma'$ by applying Lemma~\ref{lem:dist_ex}.
%As in Section~\ref{sec:td}, compute the graph of $\dmax_{\sigma'}(q) = \max_{p\in \sigma'} d(p, q)$
%by projecting the graph of $d$ over $\sigma\times\sigma'$,
%and take the upper envelope of the graphs of $\dmax_{\sigma'}$ for all $\sigma'\in\TDF$.
%By Lemma~\ref{lem:dist_ex}, we have an analogue of Lemma~\ref{lem:R} and thus
%a $\sigma$-constrained center can be computed in $O(m^2\alpha(m))$ time,
%where $m$ denotes the total complexity of all $\dmax_{\sigma'}$.
%Lemma~\ref{lem:dist_ex} implies that $m = O(n+h^2)$.
%

Our algorithms for computing the diameter and center are summarized in the proof of the following theorem.

\begin{theorem} \label{thm:main}
%Let $\calP$ be a polygonal domain of $h$ holes and $n$ vertices.
The $L_1$ geodesic diameter and center of $\calP$ can be computed
 in $O(n^2+h^4)$ and $O((n^4+n^2h^4)\alpha(n))$ time, respectively.
\end{theorem}
\begin{proof}
We first discuss the diameter algorithm, whose correctness follows directly from Lemma~\ref{lem:combined_cell}.

After the execution of the procedure of Lemma~\ref{lem:dist_vv} as a preprocessing,
our algorithm considers three cases for two cells $\sigma, \sigma'\in\TDF$:
(i) both are oceanic, (ii) both are coastal, or (iii) $\sigma$ is coastal and $\sigma'$ is oceanic.
In either case, we apply Lemma~\ref{lem:dist_ex}.

For case (i), we have $O(h^2)$ oceanic cells and the total complexity is
$\sum_{\sigma\in\TDM} n_\sigma = O(n+h^2)$.
Thus, the total time for case (i) is bounded by
\[ \sum_{\sigma\in \TDM} \sum_{\sigma'\in\TDM} O(n_\sigma n_{\sigma'})= \sum_{\sigma\in\TDM} O(n_\sigma (n+h^2)) = O((n+h^2)^2) = O(n^2+h^4).\]

For case (ii), we have $O(n)$ coastal cells in $\TDM$ and their total complexity
is $O(n)$ since they are all trapezoidal.
Thus, the total time for case (ii) is bounded by $O(n^2)$.

For case (iii), we fix a coastal cell $\sigma\in\TDF$ and iterate over
all oceanic cells $\sigma'\in\TDM$, after an $O(n)$-time preprocessing,
as done in the proof of Lemma~\ref{lem:dist_ex}.
For each $\sigma$, we take $O(n+h^2)$ time since $\sum_{\sigma\in\TDM} n_\sigma = O(n+h^2)$. Thus, the total time for case (iii) is bounded by $O(n^2 + nh^2) = O((n+h^2)^2) = O(n^2+h^4)$.

Next, we discuss our algorithm for computing a geodesic center of $\calP$.
We consider $O(n^2)$ cells $\sigma\in\TD$ and compute
all the $\sigma$-constrained centers.
As a preprocessing, we spend $O(n^4)$ time to compute
the geodesic distances $d(v,v')$ for all pairs of vertices of $\TD$
by Lemma~\ref{lem:findcells}.
Fix a cell $\sigma \in \TD$.
For all $\sigma' \in \TDF$, we compute the geodesic distance function
$d$ restricted to $\sigma\times \sigma'$ by applying Lemma~\ref{lem:dist_ex}.
As in Section~\ref{sec:td},
compute the graph of $\dmax_{\sigma'}(q) = \max_{p\in \sigma'} d(p, q)$
by projecting the graph of $d$ over $\sigma\times\sigma'$,
and take the upper envelope of the graphs of $\dmax_{\sigma'}$ for all $\sigma'\in\TDF$.
By Lemma~\ref{lem:dist_ex},
we have an analogue of Lemma~\ref{lem:R} and thus
a $\sigma$-constrained center can be computed in $O(m^2\alpha(m))$ time,
where $m$ denotes the total complexity of all $\dmax_{\sigma'}$.
Lemma~\ref{lem:dist_ex} implies that $m = O(n+h^2)$.

%Regarding the algorithm for the geodesic center,
%the correctness also follows from Lemma~\ref{lem:combined_cell}.
For the time complexity,
note that $\sum_{\sigma \in \TDM}n_\sigma = O(n+h^2)$
and $\sum_{\sigma \in \TDF \setminus \TDM} n_\sigma = O(n)$.
Since any cell in $\TD$ is either a triangle or a trapezoid,
its complexity is $O(1)$.
Thus, for each $\sigma\in \TD$, by Lemma~\ref{lem:dist_ex},
computing a $\sigma$-constrained center takes $O((n+h^2)^2 \alpha(n))$ time,
after an $O(n^4)$-time preprocessing (Lemma~\ref{lem:findcells}).
Iterating over all $\sigma\in\TD$ takes $O(n^2(n+h^2)^2\alpha(n)) = O((n^4 + n^2h^4)\alpha(n))$ time.
\end{proof}

\section{Conclusions}
\label{sec:conclusion}
We gave efficient algorithms for computing the $L_1$ geodesic diameter and center of a polygonal domain. In particular, we exploited the extended corridor structure to make the running times depend on the number of holes in the domain (which may be much smaller than the number of vertices). It would be interesting to find further improvements to the algorithms in hopes of reducing the
worst-case running times; it would also be interesting to prove non-trivial lower bounds on the time complexities of the problems.

\subparagraph*{Acknowledgements.}
Work by S.W. Bae was supported by Basic Science Research Program through the National Research Foundation of Korea (NRF) funded by the Ministry of Science, ICT \& Future Planning (2013R1A1A1A05006927), and by the Ministry of Education (2015R1D1A1A01057220). %M.~Korman was supported in part by the ELC project (MEXT KAKENHI No. 24106008).
M. Korman is partially supported by JSPS/MEXT Grant-in-Aid for
Scientific Research Grant Numbers 12H00855 and 15H02665.
J.~Mitchell acknowledges support from the US-Israel Binational Science Foundation (grant 2010074) and the National Science Foundation (CCF-1018388, CCF-1526406).
Y. Okamoto is partially supported by JST, CREST, Foundation of Innovative Algorithms for Big Data and JSPS/MEXT Grant-in-Aid for Scientific Research Grant Numbers JP24106005, JP24700008, JP24220003, and JP15K00009. V.~Polishchuk is supported in part by Grant 2014-03476 from the Sweden's innovation agency VINNOVA and the project UTM-OK from the Swedish Transport Administration Trafikverket. H.~Wang was supported in part by the National Science Foundation (CCF-1317143).

%%
%% Bibliography
%%

%% Either use bibtex (recommended), but commented out in this sample

%\bibliography{dummybib}

%% .. or use bibitems explicitely
%\newpage
\bibliographystyle{plain}
\bibliography{refLib}
%\bibliography{reference}

%\newpage
%
%\appendix
%\section*{APPENDIX}

%\noindent {\Large \bf APPENDIX}

%\vspace{0.2in}
%\noindent
%{\bf Lemma \ref{lem:20}.}
%{\em The size of the set $C(s)$ is $O(k)$.
%}
%\vspace{0.08in}

\end{document}